\documentclass[12pt,leqno]{amsart}
\usepackage[centertags]{amsmath}
\usepackage{geometry}                
\usepackage{amsthm,amsfonts,amssymb,amsmath,oldgerm}

\usepackage{multicol}
\usepackage[english]{babel}
\usepackage{amsmath}
\usepackage{psfrag}
\usepackage{amssymb}
\usepackage{amsfonts}
\usepackage{amscd}
\usepackage{amsgen}
\usepackage{amsopn}
\usepackage{amstext}
\usepackage{amsxtra}
\usepackage{longtable}
\usepackage{graphicx}
\usepackage{color}
\definecolor{grisclair}{rgb}{0.85,0.85,0.85}
\definecolor{crimson}{rgb}{0.6,0,0}
\usepackage[colorlinks,linkcolor=crimson]{hyperref}
\usepackage{float} 
\usepackage[normalem]{ulem}

\usepackage{color}
\definecolor{rred}{rgb}{0.7,0,0.1}
\definecolor{greenrb}{rgb}{0.2,0.6,0.2}

\newcommand{\mk}{\color{black}}

\newcommand{\ds}[1]{ \displaystyle{#1}}

\newcommand{\f}{\frac}
\newcommand{\dx}{\Delta x}
\newcommand{\dep}{\Delta p}
\newcommand{\de}{\partial}

\newcommand{\De}{\Delta}

\newcommand{\M}{\mathcal{M}}

\newcommand{\mbf}[1]{\mathbf{#1}}
\newcommand{\mc}[1]{\mathcal{#1}}

\newtheorem{prop}{Proposition}[section]
\newtheorem{rem}{Remark}[section]

\title[Numerical weather prediction with primitive equation]{Numerical weather prediction in two dimensions with topography, using a finite volume method}

\author[A. Bousquet]{Arthur Bousquet}
\address[Arthur Bousquet]{Institute for Scientific Computing and Applied Mathematics, Indiana University, Bloomington, Indiana, USA.}
\email{arthbous@indiana.edu}

\author[M. Chekroun]{Micka\"el D. Chekroun}
\address[Micka\"el D. Chekroun]{Department of Mathematics, University of Hawaii at Manoa, Honolulu, HI, USA, and Department of Atmospheric and Oceanic Sciences, University of California, Los Angeles, CA, USA.}
\email{mchekroun@atmos.ucla.edu}
\email{mdchekroun@math.hawaii.edu}

\author[Y. Hong]{Youngjoon Hong}
\address[Youngjoon Hong]{Institute for Scientific Computing and Applied Mathematics, Indiana University, Bloomington, Indiana, USA.}
\email{ hongy@indiana.edu}

\author[R. Temam]{Roger Temam}
\address[Roger Temam]{Institute for Scientific Computing and Applied Mathematics, Indiana University, Bloomington, Indiana, USA.}
\email{temam@indiana.edu}
\author[J. Tribbia]{Joseph Tribbia}
\address[Joseph Tribbia]{ Climate Dynamics and Predictability (CDP) section in the Division of Climate and Global Dynamics (CGD) at the National Center for Atmospheric Research (NCAR)}
\email{tribbia@ucar.edu}

\date{\today}

\numberwithin{equation}{section} 
\numberwithin{figure}{section}
\numberwithin{table}{section}

\begin{document}

\maketitle
\tableofcontents
\newpage
\begin{abstract}
We aim to study a finite volume scheme to solve the two dimensional inviscid 
primitive equations of the atmosphere with {\mk humidity and saturation, in presence of topography and subject to physically plausible boundary conditions to the system of equations.}
{\mk In that respect, a version of a projection method is introduced to enforce the compatibility condition on the horizontal velocity field, which comes from the boundary conditions. The resulting scheme allows for a significant reduction of the errors near the topography when compared to more standard finite volume schemes.}
{\mk In the numerical simulations, we first present the associated good convergence results that are satisfied by the solutions simulated by our scheme  when compared to particular analytic solutions.}
We then {\mk report on numerical} 
experiments using realistic parameters.
{\mk  Finally, the effects of a random small-scale forcing on the velocity equation is numerically investigated.  The numerical results show that such a forcing is responsible for recurrent large-scale patterns to emerge in the temperature and velocity fields. }
\end{abstract}

\section{Introduction}
We consider the inviscid primitive equations of the atmosphere with humidity.
Besides its mathematical importance, the inviscid primitive equations are used for numerical weather predictions in geophysics.
Considerable effort and attention have been devoted to this subject over the last few decades.
The theory of the inviscid primitive equations usually does not resemble the theory of the Euler equations.
In addition, it is well-known that the inviscid primitive equations are ill-posed for any set of boundary conditions of local type; see e.g.~\cite{OS78} and \cite{TT05}.
Hence the theoretical and numerical understanding of this topic is very scarce and remains as an important open problem.

In the presence of the topography and the divergence free term in the primitive equations (conservation of mass), the numerical methods require a careful adaptation.
We propose a suitable finite volume scheme to overcome these difficulties.
In addition, we consider a compatibility condition which comes from the boundary conditions.
To enforce the compatibility condition, we introduce a version of a projection method.
We perform various numerical simulations which include  deterministic and stochastic cases. The aim of this article is to see the influence of the mountain;
starting with an unsaturated humidity, we observe that the rain appears near 
the mountain.
In addition, {\mk by incorporating an additive noise to the model formulation, it is numerically illustrated that recurrent large-scale patterns can emerge  from the combined effect of a random small-scale forcing with those of the topography.}

Since Lions, Temam, and Wang proposed new formulations of the primitive equations in \cite{LTW}, the mathematical studies of primitive equations have been developed in many different directions.
Considering the viscosity, one can find many mathematical results for the primitive equations in e.g. \cite{CT2}, \cite{KG06}, \cite{KG07}, and \cite{PTZ08}.
In the absence of viscosity, two of authors of this article have studied these equations with a set of nonlocal boundary conditions in both theoretical and computational sides; see e.g. \cite{CLRTT07}, \cite{CST}, \cite{CSTT}, \cite{CTT}, \cite{RTT05} and \cite{RTT07}.
The primitive equations with humidity have been investigated in the classical references \cite{G}, \cite{H}, \cite{H2}, and \cite{RY}.
The authors from \cite{CFTT} and \cite{CT} have proposed the problem of water vapor in presence of saturation for the simplified model.
We focus, in this article, on a numerical method for the realistic and complex model of the inviscid primitive equations.

This article is organized as follows.
In Section \ref{sec2} we present our model equations and physically plausible boundary conditions. We also discuss our projection method for the velocity, and how the pressure relates to the topography.
Then in Section \ref{sec3}, we present a finite volume method to solve numerically the model equations.
Due to the topography, the classical finite volume schemes produce errors near the topography.
To resolve this problem, we propose a new scheme which is a modified 
Godunov type method that exploits the discrete finite-volume derivatives by 
using the so-called Taylor Series Expansion Scheme (TSES) introduced in 
\cite{BGHL} and \cite{GT}.

Finally, in Section \ref{sec4}, {\mk we report on numerical results based on the 
scheme introduced in Section \ref{sec3}, in the deterministic as well as 
stochastic context.} {\mk In the deterministic setting, it is shown for 
physically plausible parameter values, how the projection method 
proposed in Section \ref{ss2.3} allows for the compatibility condition \eqref{eq2.2.7b}-\eqref{eq2.2.7c} \textemdash\, that the vertical integration of the 
horizontal component of the model's velocity field must 
satisfy \textemdash\, to be satisfied to a better numerical accuracy compared to when the projection method is not used.}
The effects of a small-scale additive noise on the (horizontal 
component) of the velocity equation is then numerically investigated. As a main 
result, it is shown that such a small-scale random forcing  can significantly 
impact the model's dynamics at the large scales, leading to the appearance of 
waves that although evolving irregularly in time, manifest characteristic  
frequencies  across a low-frequency band which is more pronounced in the 
temperature field than in the velocity field.
 
 \section{The primitive equations with humidity and saturation}
\label{sec2}
Our goal in this Section is to introduce the primitive equations of the atmosphere 
with humidity and saturation, then describe the boundary conditions of the 
problems. 

The two dimensional inviscid model, in coordinates $(x,p)$, under consideration  accounts
for the conservation of horizontal momentum (in one direction), conservation of
mass and energy. The hydrostatic equation is introduced as well as the equation
of conservation of water vapor with saturation; see e.g. \cite{CBT14}, \cite{CFTT}, \cite{H},
\cite{H2}, \cite{Po}, \cite{RY}, \cite{LTW}, and \cite{PTZ08}. In the simple
humidity model that we consider, following \cite{H}, \cite{H2}, and \cite{RY}, the
water vapor leaves the system when it condenses.

\subsection{The model equations}
\label{sec2.1}
We consider the two-dimensional inviscid primitive equations which depends on two spatial like 
variables, $x \in [0,L]$ and $p \in [p_A,p_B]$, the pressure. We denote by $\M$ the (pseudo) 
spatial domain $\M=(0,L)\times(p_A,p_B)$; the function $p_B=p_B(x,t)$ refers to the pressure 
at the bottom of the atmosphere (if we do not have topography, we simply set $p_B(x,t)=p_0=1000$), 
and $p_A$ refers to the pressure at the top of the atmosphere. We choose the value $p_A = 200$. The two-dimensional inviscid primitive 
equations then read:
\begin{equation}
\label{eq1.1.1}
\left\{ 
\begin{array}{l}
\displaystyle{\frac{\de T}{\de t}+u\frac{\de T}{\de x}+\omega \frac{\de T}{\de p}=\frac{\omega}{p} \left( \frac{RT}{C_p} -\delta \frac{LF}{C_p}\right) },\\
\displaystyle{\frac{\de q}{\de t}+u\frac{\de q}{\de x}+\omega \frac{\de q}{\de p}=\delta\frac{F}{p}\omega },\\
\ds{\frac{\de u}{\de t}+u\frac{\de u}{\de x}+\omega\frac{\de u}{\de p} +\phi_x=0,}\\
\ds{\frac{\de \omega}{\de p}+\frac{\de u}{\de x}=0,}\\
\ds{\frac{\de \phi}{\de p}=-\frac{RT}{p},}\\
\ds{\phi=zg,~~z=z(x,p,t).}
\end{array} \right.
\end{equation}
 with
 \begin{itemize}
  \item $\delta=H(- \omega) H(q-q_s)$, where $H$ is the Heaviside function $H(x)
= \frac{1}{2}(1 + sign(x))$, based on equation (9.14) in \cite{H}; see also
\cite{CT}, \cite{CFTT}. Note that in the first equation \eqref{eq1.1.1}, $-H(-\omega)\omega=\omega^-=\max(-\omega,0)$
does not have a jump at $\omega = 0$.
 \item $L(T)=2.5008\times10^6-2.3\times10^3 (T-275)$ $J~kg^{-1}$ is the latent heat of vaporization (see (A4.9) in \cite{G}).
 \item $R=287~J~K^{-1}~kg^{-1}$ is the gas constant for dry air (p. 597 in \cite{G}).
 \item $R_v=461.50~J~K^{-1}~kg^{-1}$ is the gas constant for water vapor (p. 597 in \cite{G}).
 \item $p_A \in [0,200]$, usually $\simeq 200$ (chosen).
 \item $p_0=1000$ (chosen).
 \item $C_p=1004~J~K^{-1}~kg^{-1}$ is the specific heat of dry air at constant pressure (p. 475 in \cite{H}).
 \item $F(T,p)$ is given by equation (9.13) in \cite{H}:
\begin{equation}
 F(T,p)=q_s(T,p)T\left(\frac{LR-C_pR_vT}{C_pR_vT^2 +q_s(T,p)L^2(T)}\right).
\end{equation}
\item  $q_s(T,p) $ is the saturation specific humidity. From equation (9.6) in \cite{H} we have:
\begin{equation}\label{qs}
 q_s(T,p)=\frac{0.622e_s(T)}{p},
\end{equation}
 where $e_s(T)$ is the saturation vapor pressure. We approximate its value with equation (2.17) in \cite{RY}:
\begin{equation}\label{eq_es}
 e_s(T)=6.112\exp \left(\frac{17.67 (T-273.15)}{T-29.65} \right).
\end{equation}
\end{itemize}
Note that $\eqref{eq1.1.1}_1$, $\eqref{eq1.1.1}_3$ and $\eqref{eq1.1.1}_4$ express the conservation of energy, momentum in the $x$ direction and mass, respectively.\\
We aim to find:
\begin{itemize}
 \item $T=T(x,p,t)$: local temperature. 
  \item $q=q(x,p,t)$: specific humidity.
\item $u=u(x,p,t)$: the velocity along the $x$ axis.
\item $\omega=\omega(x,p,t)$: the vertical velocity in the $(x,p)$ system 
$(\omega =\frac{dp}{dt})$.
 \end{itemize}
The variables $\omega$ is a diagnostic variables which will be 
computed 
using the prognostic variables $u$. We treat the geopotential $\phi=\phi(x,p,t)$ separately using $\eqref{eq1.1.1}_5$.\\
\\
\noindent
\textbf{The boundary conditions.}
We supplement these equations with the physically relevant boundary conditions. At the top $(p = p_A)$ and bottom $(p=p_B)$ we consider the impermeable boundary conditions:
\begin{equation}\label{e:wtop}
	      \omega =0, \text{  at  } p=p_A,
\end{equation}
\begin{equation}
	      (u,\omega)\cdot\mathbf n=0, \text{  at  } p = p_B,
\label{e:impermeable}
\end{equation}
and we do not assign any boundary condition for $T$ and $q$, at $p=p_A,p_B$.
Since $\phi_x=z_xg$, from equation $\eqref{eq1.1.1}_6$, we assume that at $p=p_A$
\begin{equation}
 \phi_x=0.
\label{e:phi_pa}
\end{equation}
At $x=0$ we consider Dirichlet boundary conditions for $T$, $q$ and $u$. For $\omega$ we impose homogeneous Neumann boundary condition. At $x=L$ we consider the homogeneous Neumann boundary conditions for $T$, $q$, $u$ and $\omega$, that is,
\begin{equation} \label{e:bc_1st}
\begin{array}{ll}
\vspace{0.1in}
T \Big |_{x=0} = g_T(p), & \quad \ds{\frac{\de T}{\de x} \Big |_{x=L} = 0},\\
\vspace{0.1in}
 q \Big |_{x=0} = g_q(p), & \quad \ds{\frac{\de q}{\de x} \Big |_{x=L} = 0},\\
 \vspace{0.1in}
 u \Big |_{x=0} = g_u(p), & \quad \ds{\frac{\de u}{\de x} \Big |_{x=L} = 0},\\
 \vspace{0.1in}
 \ds{\frac{\de \omega}{\de x} \Big |_{x=0} = 0}, & \ds{\quad \frac{\de \omega}{\de x} \Big |_{x=L} = 0},
\end{array}
\end{equation}
where $g_T$, $g_q$, and $g_u$ are sufficiently smooth functions defined on $[p_A,p_B(0)]$.
We also assume the following conditions on $z_B$ and $p_B$ to make our numerical computations easier:
\begin{equation}
	  \left.\frac{\de z_B}{\de x}\right|_{x=0}=\left.\frac{\de z_B}{\de x}\right|_{x=L}=\left.\frac{\de p_B}{\de x}\right|_{x=0}=\left.\frac{\de p_B}{\de x}\right|_{x=L}=0.
\label{eq_topo}
\end{equation}
This means that the topography is flat near $x=0$ and $L$ (see e.g. Figure \ref{topo1}).

The boundary conditions \eqref{e:wtop} and \eqref{e:impermeable} plays an important role in the computation of $\omega$ from $u$.
Indeed, with the boundary condition $\omega(p_A)=0$ and $\eqref{eq1.1.1}_4$, we obtain
\begin{equation}
	  \label{eq1.1.3}
	  \omega=-\int_{p_A}^p\frac{\de u}{\de x}dp.
\end{equation}

\begin{figure}[h]
\centering \includegraphics[scale=1]{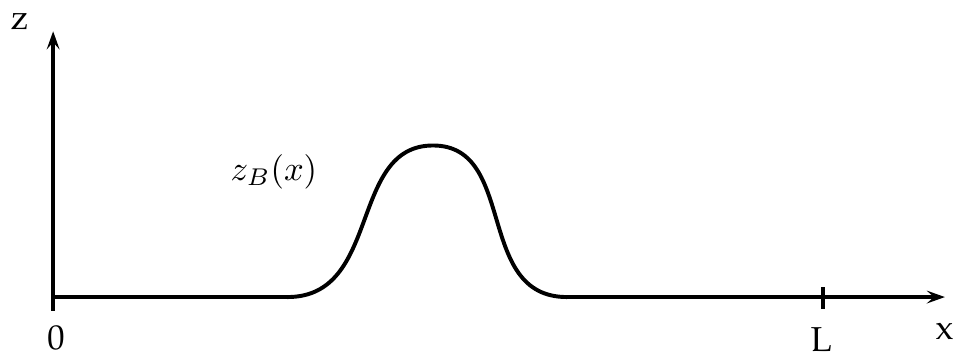}
\caption{Example of topography}
\label{topo1}
\end{figure}
\begin{rem}
\label{r1.1}
We assume that $\omega=0$ at the top ($p=p_A$), but, in general, 
$\omega(p_B)\neq 0$ in the presence of topography.
We have two boundary conditions for $\omega$: at top \eqref{e:wtop} and 
at the bottom \eqref{e:impermeable}. We directly enforce the boundary 
condition at the top by computing $\omega$ in \eqref{eq1.1.3}. We also need to 
verify that the equation \eqref{eq1.1.3}  at $p=p_B$ will satisfy 
\eqref{e:impermeable}, for this reason we propose a compatibility condition for 
$u$.

We first rewrite equation \eqref{e:impermeable} as 
\begin{equation}
\label{eq2.2.7}
	  \omega(p_B)= u(p_B)\frac{\de p_B}{\de x}.
\end{equation}
In view of \eqref{eq2.2.7} and \eqref{eq1.1.3} we deduce that
\begin{equation}
\begin{split}
	  \frac{\de}{\de x}\left(\int_{p_A}^{p_B}udp\right)& =\int_{p_A}^{p_B} \frac{\de
	  u}{\de x} dp+u(x,p_B(x),t)\frac{\de p_B}{\de x}\\
	  & = \int_{p_A}^{p_B} \frac{\de u}{\de x} dp+\omega(p_B) = 0 \quad (\text{by }\eqref{eq1.1.3}).
\end{split}
\label{eq2.2.7b}
\end{equation}
Thus, we see that $u$ satisfies the following compatibility condition:
\begin{equation}
\label{eq2.2.7c} 
	  \frac{\de}{\de x}\left(\int_{p_A}^{p_B}udp\right)=0.
\end{equation}
\end{rem}
\noindent
We recall our system of equations and we interpret $\eqref{eq1.1.1}_6$ as 
an equation for $z$ and $p$, where $z\in[z_B(x),z_A]$,
\begin{equation}
\label{eq1.1.4}
	  \phi(x,p,t)=z(x,p,t)g.
\end{equation}
\noindent
If we know $\phi$ and one of the characteristics of the topography, either $z_B(x)$ or $p_B(x)$, we can find the other one using \eqref{eq1.1.4}.
For the computation of $\phi_x$, we  differentiate in $x$ equation $\eqref{eq1.1.1}_5$ and we obtain
\begin{equation}
	  \phi_{xp}=-\frac{RT_x}{p},
 \label{e:phi_2}
\end{equation}
then, from the boundary condition \eqref{e:phi_pa}, we deduce
\begin{equation}
	  \phi_{x}=-\int_{p_A}^p\frac{RT_x}{p}.
 \label{e:phi_3}
\end{equation}

\subsection{Reformulation of the equations}
\label{sec2.2}
In this subsection we rewrite and simplify \eqref{eq1.1.1} in view of the numerical simulations in Sections \ref{sec3} and \ref{sec4}.
Thanks to the divergence free condition $\eqref{eq1.1.1}_4$, the equations of $T$, $q$ and $u$
in $\eqref{eq1.1.1}_1$-$\eqref{eq1.1.1}_3$ read:
\begin{equation}
\begin{cases}
\label{rewrite1}
	  \displaystyle{\frac{\de T}{\de t}+\nabla_{x,p}(uT,\omega T)=\frac{\omega}{p} \left( \frac{RT}{C_p} -\delta \frac{LF}{C_p}\right)},\\
	  \displaystyle{\frac{\de q}{\de t}+\nabla_{x,p}(uq,\omega
	  q)=\delta~\frac{F}{p}~\omega},\\
	  \ds{\frac{\de u}{\de t}+\nabla_{x,p}(uu,\omega u) +\phi_x=0.}\\
\end{cases}
\end{equation}
We assume that $T$ and $\phi$ are perturbations of a stratified configuration, ($\bar T(p)$, $\bar \phi(p)$) satisfying the hydrostatic equation
\begin{equation}
\label{phibar}
	  \frac{\de\bar \phi}{\de p}=-\frac{R\bar T}{p}.
\end{equation}
Noting that $z_B(x)g = \phi(x,p_B(x,t),t)$, we deduce 
\begin{equation}
	  z_B(x)g \simeq \bar\phi(p_B(x,t)),
\end{equation}
and this implies that $p_B$ does not depend on $t$ as announced, that is:
\begin{equation}
	  p_B(x,t) = p_B(x).
\label{pb1}
\end{equation}
For the sake of simplicity, we set
\begin{equation}
\begin{split}
	  &\mbf u = (T,q,u),\enspace \mbf S=\left(\frac{\omega}{p} \left( \frac{RT}{C_p} 
	  -\delta \frac{LF}{C_p}\right),\delta\frac{F}{p}\omega,0\right),\\
	  & \mbf \Phi_x =(0,0,\phi_x), \enspace \mbf G=(g_T,g_q,g_u).
	  \end{split}
	  \label{e:notation}
\end{equation}

In view of  of the notations above, \eqref{eq1.1.3}, \eqref{e:phi_3}, and 
\eqref{rewrite1} we arrive at the following boundary value problem:
\begin{equation}
\begin{cases}
	  \displaystyle{\frac{\de \mbf u}{\de t}+\nabla_{x,p}(u\mbf u,\omega 
	  \mbf u)+\mbf \Phi_x=\mbf S},\\
	  \ds{\omega=-\int_{p_A}^p\frac{\de u}{\de x}dp,}\\
	  \ds{\phi_x=-\int_{p_A}^p\frac{RT_x}{p},}
\end{cases}
\label{bvp}
\end{equation}
and the boundary conditions are the same as {in \eqref{e:bc_1st}}:
\begin{equation}
\label{bc}
\begin{split}
	  & \frac{\de z_B}{\de x}=\frac{\de p_B}{\de x}=0, \text{  at  } x=\{0,L\},\\
	  & \omega = u\frac{\de p_B}{\de x},\text{  at  } p=p_B, \qquad
	    \omega=\phi_x=0, \text{  at  } p=p_A,\\
	  & \mbf u = \mbf G(p), \text{  at  } x=0, \qquad
	    \frac{\de \mbf u}{\de n} = 0, \text{  at  } x=L, \\
	  & \frac{\de \omega}{\de n}=0, \text{  at  } x=\{0,L\}.
\end{split}
\end{equation}

\subsection{Projection of $u$}
\label{ss2.3}
In this section, we develop a projection method to ensure the compatibility condition 
\eqref{eq2.2.7c} in Remark \ref{r1.1}; as we will see, this projection method is similar -- but 
simpler -- than the projection method in incompressible fluid mechanics Temam \cite{T} and Chorin \cite{CA}.
Let $\tilde{u} \in L^2(\M)$ be the solution in $\eqref{bvp}_3$ that does not satisfy
\eqref{eq2.2.7c}. We now construct its projection $u= P \tilde u$ on the appropriate set of 
functions that satisfy the compatibility condition (\ref{eq2.2.7c}):
\begin{equation}
 PL^2=\left\{ v \in L^2(\M),\enspace \frac{\de}{\de x}\left(\int_{p_A}^{p_B}vdp\right)=0 \right\}.
\end{equation}
We denote by $(I-P)L^2$ the orthogonal complement of $PL^2$ in $L^2(\M)$.
\begin{prop}
The orthogonal complement of $PL^2$ in $L^2(\M)$ can be characterized as follows
\begin{equation}\label{e2.14}
(I-P)L^2=\left\{\alpha, \enspace \alpha=\alpha(x) \in L^2([0,L]), \int_0^L\alpha dx=0\right\}.
\end{equation}
\end{prop}
\begin{proof}
We first show that $(I-P)L^2\supset R.H.S$ of \eqref{e2.14}. Let
$\alpha=\alpha(x) \in L^2([0,L])$ be independent of $p$ and
$\int_0^L\alpha(x)dx=0$. Let $\bar \alpha$ be the primitive function of
$\alpha$ vanishing at $x=0$, that is,
\begin{equation}
\bar\alpha(x) = \int_0^x\alpha(x')dx',
\end{equation}
so that $\alpha =\bar\alpha_x$ and $\bar\alpha\in H^1(0,L)$.
Note also that $\bar\alpha(L)=0$ since $\int_0^L\alpha(x)dx=0$. Then for $v \in PL^2$, we deduce that
\begin{equation*}
\begin{split}
	  (v,\alpha)_{L^2} = (v,\bar\alpha_x)_{L^2} & =\int_0^L \int_{p_A}^{p_B}v\bar\alpha_xdpdx=\int_0^L\bar\alpha_x\left(\int_{p_A}^{p_B}v dp\right)dx
	  \\
	  & = \left. \bar\alpha \left(\int_{p_A}^{p_B}v dp\right) \right|_0^L - \int_0^L \bar\alpha \frac{\de}{\de x}\left(\int_{p_A}^{p_B}v dp\right)dx \\
	  & = 0. 
\end{split}
\end{equation*}
This implies that $\left\{ \alpha, \enspace \alpha=\alpha(x) \in L^2([0,L]), \int_0^L\alpha dx=0\right\}\subset (I-P)L^2$.\\
Conversely, we prove that $(I-P)L^2\subset R.H.S$ of \eqref{e2.14}. Let $\tilde u\in L^2(\M)$ and let $u$ be its orthogonal projection on $PL^2$, then $\tilde u-u=(I-P)\tilde u$ and 
$$(u-\tilde u,v)=0,\enspace \forall v \in PL^2,$$
which in fact characterizes $u$.
For  $\phi \in C^\infty_c(\M)$ we have
$$\frac{\de}{\de x}\left(\int_{p_A}^{p_B}\frac{\de \phi}{\de p}dp\right)= \frac{\de}{\de x} \left(\phi(p_B,x)-\phi(p_A,x)\right)=0,$$
so $\de \phi/\de p \in PL^2$, which gives 
$$(u-\tilde{u},\frac{\de \phi}{\de p})=0,\enspace \forall\phi \in C^\infty_c(\M).$$
This implies that $\frac{\de}{\de p}(u-\tilde{u})=0$
in the sense of distributions on $\M$. Therefore $u-\tilde{u}$ does not depend on $p$, that is, $u-\tilde u=\alpha\in L^2(0,L)$ (see \cite{Sw} and Section 4.4 in \cite{PTZ08} for more details about distributions independent of one variable). It remains to prove that $ \int_0^L \alpha(x) dx=0$, but for $v \in PL^2$
\begin{equation*}
\begin{split}
	  & \int_0^L\int_{p_A}^{p_B} (u-\tilde u)vdpdx=0,\\
	  & \Longrightarrow \int_0^L (u-\tilde u)\int_{p_A}^{p_B}v dp dx =0,\\
	  & \Longrightarrow \int_0^L (u-\tilde u) dx =0.
\end{split}
\end{equation*}
For the last implication we have chosen an arbitrary $v \in PL^2$ such that $\int ^{p_B}_{p_A} v dp$ 
-- which is constant in $x$ and independent of $p$ -- is not zero. Thus, we have 
\begin{equation} \notag
\left\{ \alpha, \enspace \alpha=\alpha(x) \in L^2([0,L]), \int_0^L\alpha dx=0\right\}\supset (I-P)L^2,
\end{equation}
and \eqref{e2.14} is proved.
\end{proof}
Now we assume that $\phi$ and $\omega$ are known and denote the third component of the solution of equation $\eqref{bvp}_1$ by $\tilde u \in L^2(\M)$;
so far $\omega$ is not required yet to satisfy $\eqref{bvp}_2$.
Let $u \in PL^2$ be the orthogonal projection of $\tilde u$ onto $PL^2$ and let $\lambda_x(x)$ be its orthogonal complement so that $\lambda, \lambda_x \in L^2(0,L)$.
Hence, we obtain
\begin{equation}
 \label{eq2.3.2}
u+\lambda_x=\tilde{u}.
\end{equation}
We integrate \eqref{eq2.3.2} in $p$ from $p_A$ to $p_B$ and differentiate in $x$ using \eqref{eq2.2.7c}:
\begin{equation}
 \label{2.3.3}
	  \frac{\de}{\de x}\int_{p_A}^{p_B}\lambda_x(x)dp=\frac{\de}{\de x}\int_{p_A}^{p_B}\tilde{u}dp.
\end{equation}
From \eqref{2.3.3} we deduce the equation for $\lambda_x$:
\begin{equation}
 \label{2.3.4}
	  \frac{\de}{\de x}\left((p_B-p_A)\lambda_x\right)=\frac{\de p_B}{\de x}\lambda_{x}+(p_B-p_A)\lambda_{xx}=\frac{\de}{\de x}\int_{p_A}^{p_B}\tilde{u}dp.
\end{equation}

\subsection{Treatment of $p_B(x)$}
\label{sec2.4}
As explained in Section \ref{sec2.1}, the topography is determined by either $p_B(x)$ or $z_B(x)$. To find the relation between $z_B$ and $p_B$, we use $\eqref{eq1.1.1}_6$ and $\eqref{phibar}$:
\begin{equation}
 \label{eq2.4.1}
	    \frac{\de z}{\de p}g=\frac{\de \bar{\phi}}{\de p}=-\frac{R\bar{T}}{p} = -\frac{R}{p}[T_0 - (1 - \frac{p}{p_0})\Delta T],
\end{equation}
then by integrating in $p$
\begin{equation}
 \label{eq2.4.2}
  zg=-R(T_0-\Delta T)\ln(p)-\frac{R\Delta T}{p_0}p+C.
\end{equation}
At $z=0$, we have $p=p_0$ (virtual pressure on the whole segment) and this gives C:
$$C=R(T_0-\Delta T)\ln(p_0)+R\Delta T.$$
Moreover, from the fact that $z=z_B$ at $p=p_B$, we deduce that
\begin{equation}
 \label{eq2.4.3}
z_B(x)g=-R(T_0-\Delta T)\left(\ln(p_B(x))-\ln(p_0)\right)-\frac{R\Delta T}{p_0}p_B(x)+R\Delta T.
\end{equation}
Since the topography is known, that is the function $z_B=z_B(x)$ is given, and we can then compute $p_B(x)$ from \eqref{eq2.4.3} using the Newton method. However, in our simplified calculations, we choose $p_B(x)$ and deduced the topography from \eqref{eq2.4.3} to avoid the repeated use of the Newton method.

\section{Numerical scheme: the finite volume method}
\label{sec3}
In this article, we use the Godunov's method as {\mk described} in Chapter 23 of \cite{L} combined with a spatial discretization by finite volumes  {\mk which has to be performed with a special care here due to the topography.} In the presence of topography, the spatial domain in $x$ and $p$ is not rectangular.
Such a geometry gives rise to computational difficulties since the classical methods for the directional derivatives are not accurate.
To avoid this problem, we first propose a specific discretization in a given spatial domain in Section \ref{sec3.1}. Then, we introduce the first order finite volume scheme to compute $T$, $q$ and $\tilde u$ in Section \ref{sec3.2}. We then study the discrete projection method, in Section \ref{sec:prom}, to obtain the solution $u$.
Following the computation of the diagnostic variables, we consider the computation of $\omega$, in Section \ref{sec:omega}, and the computation of $\phi$ in Section \ref{sec:phi}, the two prognostic variables.
Finally, in Section \ref{sec3.3} we
present our discretization scheme in time which is the classical Runge-Kutta 4th-order method.
\subsection{Space discretization}
\label{sec3.1}
Let us set the spatial domain $\M=[0,L]\times[p_A,p_B(x)]$ which contains a topography. On the interval $[0,L]$ in the $x$-direction, we define
\begin{equation}
 \begin{cases}
	    x_{i-\frac{1}{2}} = (i-1) \dx ,\quad  1 \leq i \leq N_x +1,\\
	    0= x_{\frac{1}{2}} < x_{\frac{3}{2}} < ... < x_{N_x + \frac{1}{2}} = L,
 \end{cases}
\end{equation}
where $\dx = \frac{L}{N_x}$. For fixed $x=x_{i-\f{1}{2}}$, $1 \leq i \leq N_x+1$, we define
\begin{equation}
 \begin{cases}
	    p_{i-\f{1}{2}, j-\f{1}{2}} = p_A + (j-1) \dep_{i-\f{1}{2}},\\
	    p_A= p_{i-\f{1}{2},\f{1}{2}} < p_{i-\f{1}{2},\f{3}{2}} < ... < p_{i-\f{1}{2},N_p+\f{1}{2}} = p_B(x_{i-\f{1}{2}}),
 \end{cases}
\end{equation}
where $\dep_{i-\f{1}{2}} = \dfrac{p_B(x_{i-\f{1}{2}})-p_A}{N_p}$.
We then discretize the domain $\M$ in $(N_x+2) \times (N_p+2)$ cells $C_{i,j}$ where $0 \leq i \leq N_x+1$ and $0 \leq j \leq N_p+1$. 
For $1 \leq i \leq N_x$ and $1 \leq j \leq N_p$, the cells $C_{i,j}$ are trapezoid; see e.g. Figure \ref{mesh1}.
For $i=0, N_x+1$ or $j=0,N_p+1$, the cells $C_{i,j}$ refer to the flat control volumes on the boundary of $\M$. 
\noindent
For the inside cells, we set for $1 \leq i \leq N_x$ and $1 \leq j \leq N_p$
\begin{equation}
	  C_{i,j} := \text{trapezoid connecting  } \mbf{x}_{i-\f{1}{2},j-\f{1}{2}}, \mbf{x}_{i-\f{1}{2},j+\f{1}{2}}, \mbf{x}_{i+\f{1}{2},j+\f{1}{2}}, \text{  and  } \mbf{x}_{i+\f{1}{2},j-\f{1}{2}}
\end{equation}
where $\mbf{x}_{i-\f{1}{2},j-\f{1}{2}} = (x_{i-\f{1}{2}},p_{i-\f{1}{2},j-\f{1}{2}})$.

We now consider the barycenter of the inside cells (see Figure \ref{f_bary}).
Using the diagonals, we split the quadrilateral cell $C_{i,j}$, for $1 \leq i \leq N_x$ and $1 \leq j \leq N_p$, into four different triangles and find the barycenter of each of them:
\begin{equation}
\begin{split}
	  & \bar x_1 = \frac{x_{i-\frac{1}{2}} + x_{i+\frac{1}{2}} + x_{i-\frac{1}{2}}}{3},
	      \quad \bar p_1 = \frac{p_{i-\frac{1}{2}, j-\frac{1}{2}} + p_{i+\frac{1}{2}, j-\frac{1}{2}} + p_{i-\frac{1}{2}, j+\frac{1}{2}}}{3} , \\
	  & \bar x_2 = \frac{x_{i-\frac{1}{2}} + x_{i+\frac{1}{2}} + x_{i+\frac{1}{2}}}{3},
	      \quad \bar p_2 = \frac{p_{i-\frac{1}{2}, j-\frac{1}{2}} + p_{i+\frac{1}{2}, j-\frac{1}{2}} + p_{i+\frac{1}{2}, j+\frac{1}{2}}}{3} , \\
	  & \bar x_3 = \frac{x_{i+\frac{1}{2}} + x_{i+\frac{1}{2}} + x_{i-\frac{1}{2}}}{3},
	      \quad \bar p_3 = \frac{p_{i+\frac{1}{2}, j-\frac{1}{2}} + p_{i+\frac{1}{2}, j+\frac{1}{2}} + p_{i-\frac{1}{2}, j+\frac{1}{2}}}{3} , \\
	  & \bar x_4 = \frac{x_{i-\frac{1}{2}} + x_{i+\frac{1}{2}} + x_{i -\frac{1}{2}}}{3},
	      \quad \bar p_4 = \frac{p_{i-\frac{1}{2}, j-\frac{1}{2}} + p_{i+\frac{1}{2}, j+\frac{1}{2}} + p_{i -\frac{1}{2}, j+\frac{1}{2}}}{3}.
\end{split}
\end{equation}
\begin{figure}[H]
	  \centering
	  \includegraphics[scale=0.5]{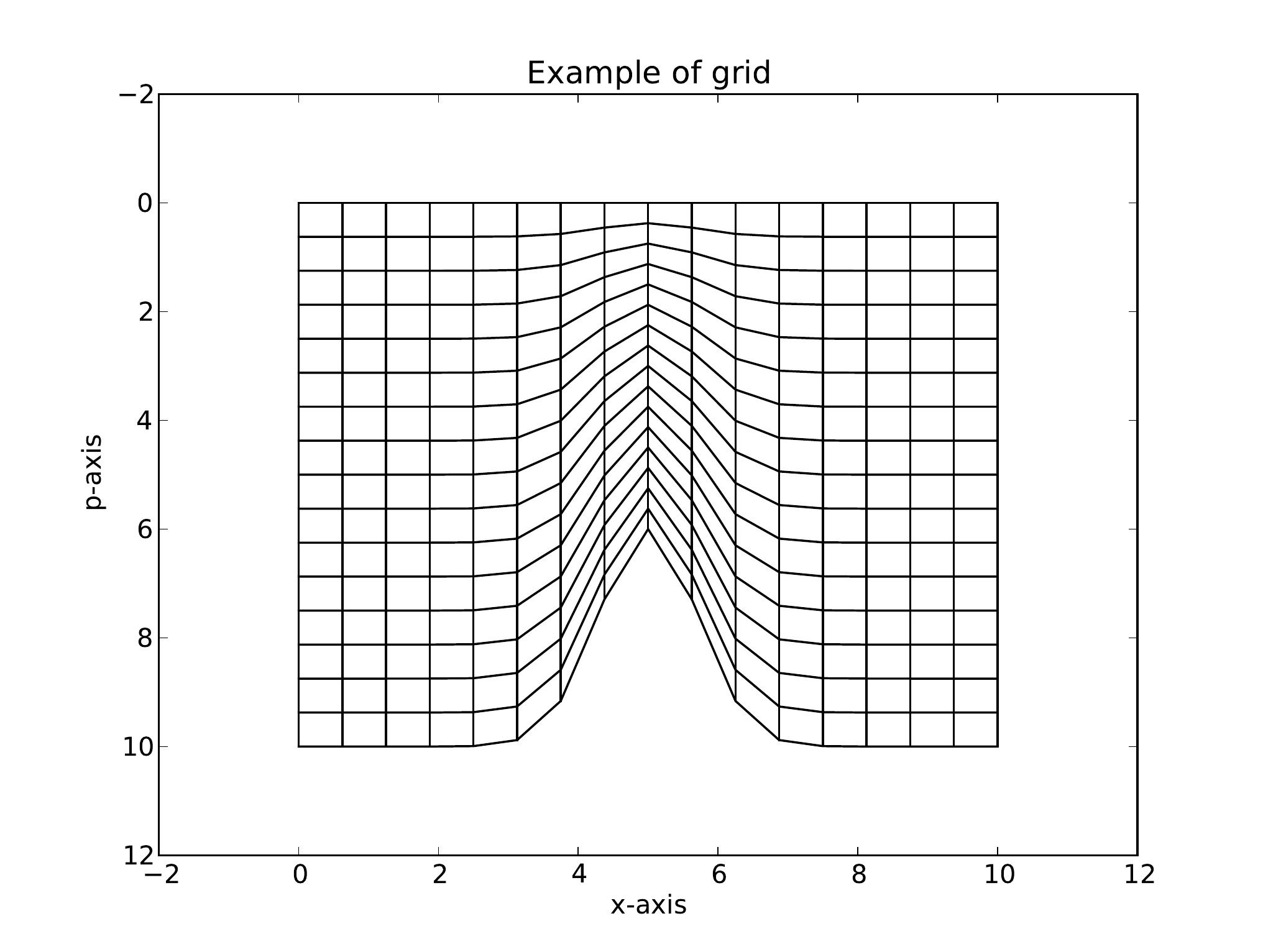}
	  \caption{Example of the spatial discretization}
\label{mesh1}
\end{figure}
The barycenter of the quadrilateral cell $C_{i,j}$ is the point of intersection between two lines which pass through $\bar x_1$ and $\bar x_3$, and $\bar x_2$ and $\bar x_4$, respectively.
Thus, the barycenter $(x_m,y_m)$ is
\begin{equation}
\begin{split}
	  & x_m = \frac{\bar A_1 x_1 - \bar A_2 \bar x_2 + \bar p_2 - \bar p_1 }{\bar A_1 - \bar A_2}, \\
	  & p_m = \bar A_1(x_m - \bar x_1) + \bar p_1,
\end{split}
\end{equation}
where 
$\bar A_1 = \frac{\bar p_1 - \bar p_3}{\bar x_1 - \bar x_3}$ and $\bar A_1 = \frac{\bar p_2 - \bar p_4}{\bar x_2 - \bar x_4}$.
\begin{figure}[h!]
	  \centering
	  \includegraphics[scale=1]{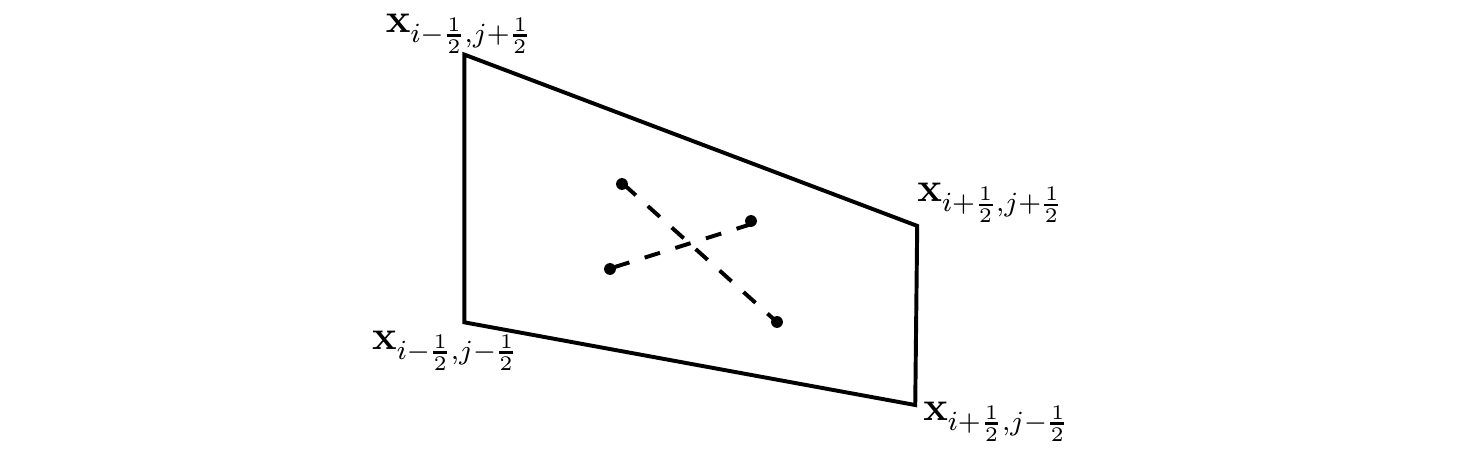}
	  \caption{Computation of the barycenter of a quadrilateral cell.}
	  \label{f_bary}
\end{figure}
We also define the centers of the East and West edges to compute the
fluxes in Section \ref{sec3.2} below.
For a given cell $C_{ij}$, let us call the centers of the West edge as 
  $(x_{i-\frac{1}{2}}, p_{i-\frac{1}{2},j})$
and the centers of the East edge as
  $(x_{i+\frac{1}{2}}, p_{i+\frac{1}{2},j})$
  , respectively.
Then, the centers of the edges are defined below:                                
\begin{equation}
\begin{split}
	  & (x_{i-\frac{1}{2}}, p_{i-\frac{1}{2},j}) 
		    =  \Big(x_{i-\frac{1}{2}},\frac{p_{i-\frac{1}{2},
		    	j-\frac{1}{2}} + p_{i-\frac{1}{2},j+\frac{1}{2}}}{2} \Big),\\
	  & (x_{i+\frac{1}{2}}, p_{i+\frac{1}{2},j}) 
		    = \Big ( x_{i+\frac{1}{2}},\frac{p_{i+\frac{1}{2},
		    	j-\frac{1}{2}} + p_{i+\frac{1}{2},j+\frac{1}{2}}}{2} \Big ).	
\end{split}
\end{equation}	
We can find the North/South center of the edges in the same way. Note that the vertical edges of the cells are parallel to the $p$-axis so that the barycenter of the trapezoidal cells are well-aligned in the $p$-direction; See Figures \ref{mesh1} and \ref{fig.4.2.1}.\\
We introduce the flat control volumes along the boundary of $\M$ to impose the boundary conditions:
\begin{equation}
        \left\{
                \begin{array}{rl}
                        C_{0, j}
& 
                               = \text{
                                                segment joining
                                                $\mathbf{x}_{  \frac{1}{2}, \, j-\frac{1}{2}}$ and
$\mathbf{x}_{\frac{1}{2}, \, j+\frac{1}{2}}$},
                        \quad
                        1 \leq j \leq N_p.
                        \\
                        C_{N_x+1, j}
& 
                               = \text{
                                                segment joining
                                                $\mathbf{x}_{ N_x + \frac{1}{2}, \, j-\frac{1}{2}}$
and $\mathbf{x}_{ N_x + \frac{1}{2}, \, j+\frac{1}{2}}$,
                                }
                        \quad
                        1 \leq j \leq N_p.
			\\
                        C_{i, 0}
& 
                               = \text{
                                                segment joining $\mbf{x}_{i- \frac{1}{2},
\frac{1}{2}}$ and $\mbf{x}_{i+\frac{1}{2}, \frac{1}{2}}$,
                                                }
                        \quad
                        1 \leq i \leq N_x,\\
                    
                        C_{i, N_p+1}
& 
                               = \text{
                                                segment joining $\mbf{x}_{i-\frac{1}{2},
N_p+\frac{1}{2}}$ and $\mbf{x}_{i+\frac{1}{2}, N_p+\frac{1}{2}}$,
                                                }
                        \quad
                        1 \leq i \leq N_x.
                \end{array}
        \right.
\label{e:K_i,j_fic}
\end{equation}
We set the centers of the flat control volumes in \eqref{e:K_i,j_fic} as follows:
\begin{equation}\label{e:barycenter_flat}
        \left\{
                \begin{array}{l}
                    
                        \mbf{x}_{i, 0}
                                = \dfrac{1}{2}
                                    \Big(
                                        \mbf{x}_{i- \frac{1}{2}, \frac{1}{2}}
                                        +
                                        \mbf{x}_{i+\frac{1}{2}, \frac{1}{2}}
                                    \Big),
                        \quad
                        \mbf{x}_{i, N_p+1}
                                    = \dfrac{1}{2}
                                        \Big(
                                                \mbf{x}_{i-\frac{1}{2}, N_p+\frac{1}{2}}
                                                +
                                                \mbf{x}_{i+\frac{1}{2}, N_p+\frac{1}{2}}
                                        \Big),
                        \quad
                        1 \leq i \leq N_x,\\
                         \mbf{x}_{0, j}
                                     =\dfrac{1}{2}
                                         \Big(
                                                     \mbf{x}_{\frac{1}{2}, j-\frac{1}{2}}
                                                     +
                                                    \mbf{x}_{\frac{1}{2}, j+\frac{1}{2}}
                                         \Big),
                         \quad
                         \mbf{x}_{N_x+1, j}
                                    =\dfrac{1}{2}
                                         \Big(
                                                 \mbf{x}_{N_x+\frac{1}{2}, j-\frac{1}{2}}
                                                 +
                                                 \mbf{x}_{N_x+\frac{1}{2}, j+\frac{1}{2}}
                                        \Big),
                         \quad
                         1 \leq j \leq N_p.
                 \end{array}
         \right.
\end{equation}
We also define the segments $\Gamma_{i,j+\frac{1}{2}}$ and $\Gamma_{i+\frac{1}{2},j}$
\begin{equation}
 \begin{split}
  & \Gamma_{i,j+\frac{1}{2}} \mbox{ is the segment connecting } \mbf{x}_{i-\frac{1}{2},j+\frac{1}{2}} \mbox{ and }  \mbf{x}_{i+\frac{1}{2},j+\frac{1}{2}},\\
    & \Gamma_{i+\frac{1}{2},j} \mbox{ is the segment connecting } \mbf{x}_{i+\frac{1}{2},j-\frac{1}{2}} \mbox{ and }  \mbf{x}_{i+\frac{1}{2},j+\frac{1}{2}}.\\
 \end{split}
\end{equation}
We now introduce the finite volume space $V_h$:
\begin{equation}\label{e:FV_space}
 V_{h} := \left\{ \begin{array}{l}
             \text{space of step functions $u_{h}$ on } \overline{\M} \text{ such that }\\
             {u_{h}}|_{C_{i,j}} = u_{i,j} , \text{ } 0 \leq i \leq N_x+1, \text{ } 0 \leq j
 \leq N_p+1.
             \end{array}\right\}.
 \end{equation}
We then write
\begin{equation}\label{e_FV:5}
            u_{h}
                    = \sum_{i=0}^{N_x+1}\sum_{j=0}^{N_p+1} u_{i,j}\chi_{C_{i,j}},
\end{equation}
where $\chi_{C_{i,j}}$ is the characteristic function on $C_{i,j}$.
\bigskip
\bigskip
For the computation of $\omega$ in Section \ref{sec:omega}, and $\phi$ in Section \ref{sec:phi}, we construct the quadrilateral cells $C_{i,j+\frac{1}{2}}$ to employ the finite volume derivatives as in \cite{BGHL}, \cite{GT}, and \cite{GT13}:
\begin{equation}\label{quad1}
    C_{i,j+\frac{1}{2}}
                :=
                       \text{quadrilateral connecting
                                $\mathbf{x}_{i - \frac{1}{2}, \, j+\frac{1}{2}}$,
                                $\mathbf{x}_{i		    , \, j	      }$,
                                $\mathbf{x}_{i+\frac{1}{2}, \, j+\frac{1}{2}}$,
                                and
                                $\mathbf{x} _{i		    , \, j+1	      }$,
                                }
\end{equation}
for $1 \leq i \leq N_x$, $0 \leq j \leq N_p$; see Figure \ref{fig.4.2.1}.
\begin{figure}[h]
	  \centering
	  \includegraphics[scale=0.8]{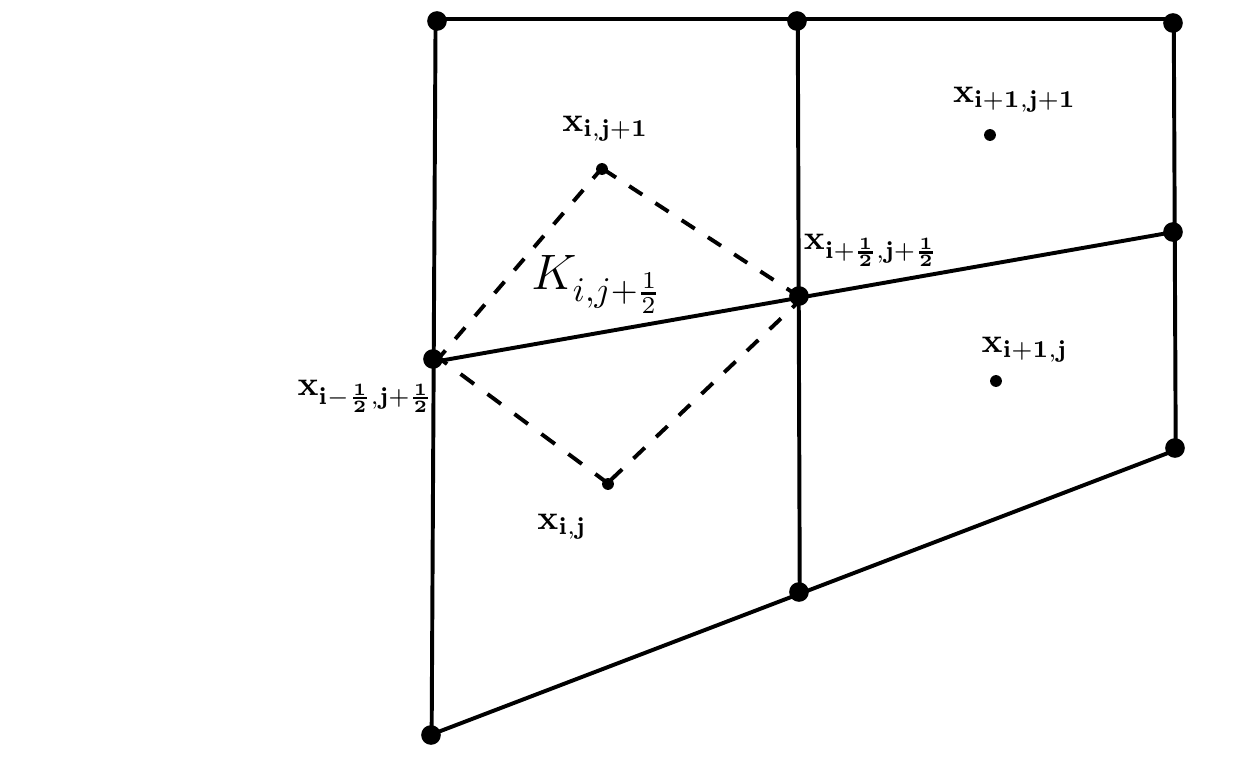}
	  \caption{Computation of the barycenter of a cell.}
	  \label{fig.4.2.1}
\end{figure}
We aim to compute the gradients of $u$, $\omega$, and $\phi$ in $C_{i,j+\frac{1}{2}}$. Let us start by defining the coefficients $a^1_{i+\frac{1}{2},j+\frac{1}{2}}$, $a^2_{i+\frac{1}{2},j+\frac{1}{2}}$, $a^3_{i+\frac{1}{2},j+\frac{1}{2}}$, and $a^4_{i+\frac{1}{2},j+\frac{1}{2}}$ satisfying the equations below, for $0\leq i \leq N_x$ and $1\leq j\leq N_p$:\\
\begin{equation}
  \begin{cases}
      {\bf x}_{i+\frac{1}{2},j+\frac{1}{2}} = a^1_{i+\frac{1}{2},j+\frac{1}{2}} {\bf x}_{i,j} + a^2_{i+\frac{1}{2},j+\frac{1}{2}} {\bf x}_{i+1,j} + a^3_{i+\frac{1}{2},j+\frac{1}{2}} {\bf x}_{i,j+1} + a^4_{i+\frac{1}{2},j+\frac{1}{2}} {\bf x}_{i+1,j+1},\\
      1 = a^1_{i+\frac{1}{2},j+\frac{1}{2}}+ a^2_{i+\frac{1}{2},j+\frac{1}{2}} +a^3_{i+\frac{1}{2},j+\frac{1}{2}}+ a^4_{i+\frac{1}{2},j+\frac{1}{2}}.
  \end{cases}\label{lin1}
\end{equation}
We then obtain $a^l_{i+\frac{1}{2},j+\frac{1}{2}}$ for $l=1,2,3,4$ by fixing one of the four variables; see \cite{BGHL} for more details. We then compute $u_{i+\frac{1}{2},j+\frac{1}{2}}$:
\begin{equation}
 u_{i+\frac{1}{2},j+\frac{1}{2}} = a^1_{i+\frac{1}{2},j+\frac{1}{2}}u_{i,j} + a^2_{i+\frac{1}{2},j+\frac{1}{2}} u_{i+1,j} + a^3_{i+\frac{1}{2},j+\frac{1}{2}} u_{i,j+1} + a^4_{i+\frac{1}{2},j+\frac{1}{2}} u_{i+1,j+1}.
\end{equation}

Now, we can define the non singular matrices $M_{i,j+\frac{1}{2}}$ whose row vectors represent the diagonal of $C_{i+j+\frac{1}{2}}$:
\[
 M_{i,j+\frac{1}{2}} :=
 \begin{pmatrix}
  x_{i+\frac{1}{2},j+\frac{1}{2}} - x_{i-\frac{1}{2},j+\frac{1}{2}} & p_{i+\frac{1}{2},j+\frac{1}{2}} -p_{i-\frac{1}{2},j+\frac{1}{2}} \\
  x_{i,j+1}- x_{i,j} & p_{i,j+1} -p_{i,j}
 \end{pmatrix}.
\]
Then, we obtain the gradient $\nabla_h u_h:= (\nabla_h^x u_h,\nabla_h^p u_h)$ (or $\nabla_h \omega_h$ or $\nabla_h \phi_h$)  as follows:
\begin{equation}\label{diamond}
 \left.\nabla_h u_h \right|_{C_{i,j+\frac{1}{2}}}:=
  M^{-1}_{i,j+\frac{1}{2}} \cdot
 \begin{pmatrix}
  u_{i+\frac{1}{2},j+\frac{1}{2}} - u_{i-\frac{1}{2},j+\frac{1}{2}} \\
  u_{i,j+1} - u_{i,j}
 \end{pmatrix}\quad  \mbox{ for  } \quad 1\leq j \leq N_p-1.
\end{equation}

\subsection{Finite volume scheme: Godunov's scheme}
\label{sec3.2}
In our simulations, we derive our finite volume method from an upwind finite volume method, see e.g. 
\cite{L}, to implement our schemes. 
From \eqref{e:notation} -- \eqref{bc} we define the finite volume space for $\mbf u$
\begin{equation}\label{e:FV_space1}
	  \mc V_{h} := 
	      \left\{ \begin{array}{l}
			 \mbf u_{h}=(T_h,q_h,u_h) \in (V_h)^3 
				  \text{ such that }\\
			 \mbf u_{N_x+1,j}=\mbf u_{N_x,j}, \enspace 1 \leq j \leq N_p,\\
			  \mbf u_{0,j} 
				  = \mbf G (x_{0},p_{0,j}), \enspace  1 \leq j 
\leq N_p
		      \end{array}\right\}.
 \end{equation}
The finite volume space for $\omega$ is
 \begin{equation}\label{e:FV_space2}
		    \mc W_{h} := 
			  \left\{ \begin{array}{l}
				      \omega_{h}\in V_h \text{ such that }\\
				      \omega_{N_x+1,j}= \omega_{N_x,j},\enspace 1\leq j \leq N_p,\\
				      \omega_{0,j}=\omega_{1,j},\enspace 1\leq j \leq N_p,\\
				      \omega_{i,0} = 0, \enspace 1 \leq i \leq N_x\\
				   \end{array}\right\}.
 \end{equation}
For our last unknown $\phi_x$, the finite volume space is
 \begin{equation}\label{e:FV_space3}
		  \mc K_{h} := \left\{ \begin{array}{l}
					      \mbf (\phi_{x})_h \in V_h \text{ such that }\\
					      (\phi_{x})_{i,0} = 0,\enspace 1 \leq i \leq N_x\\
					\end{array}\right\}.
 \end{equation}
By integrating \eqref{bvp} on each cell to project \eqref{bvp} onto the finite volume spaces, we obtain
\begin{equation}
	  \frac{d \mbf u_{i,j}}{dt} 
	  + \frac{1}{|C_{i,j}|} \int_{C_{i,j}} \nabla_{x,p}(u\mbf u,\omega   \mbf u)dxdp
	  +(\mbf \Phi_x)_{i,j}
	  =\mbf S_{i,j},\enspace 1\leq i\leq N_x,\; 1\leq j \leq N_p,
	  \label{e:bvp_h}
\end{equation}
where
\begin{equation}
 \begin{split}
	  & \mbf u_{i,j}(t)
		      = \frac{1}{|C_{i,j}|}\int_{C_{i,j}} \mbf u,\\
	  & \mbf S_{i,j}(t)
		      = \frac{1}{|C_{i,j}|}\int_{C_{i,j}} \mbf S,\\
	  & (\mbf \Phi_x)_{i,j}(t)
		      = \frac{1}{|C_{i,j}|}\int_{C_{i,j}} \mbf \Phi_x,\\
	  & \omega_{i,j}(t)
		      = \frac{1}{|C_{i,j}|}\int_{C_{i,j}}\omega,
 \end{split}
\end{equation} 
and $\mbf u$, $\mbf S$, and $\mbf \Phi_x$ are as in \eqref{e:notation}. In this subsection, we focus on the fluxes and find $\mbf{u} = (T,q,u)$ using upwind schemes. We then look for $\omega$ and $\phi_x$ separately in Sections \ref{sec:omega} and \ref{sec:phi}.\\

Using the divergence theorem, we obtain that
\begin{equation}\begin{split}
	  &\frac{1}{|C_{i,j}|}\int_{C_{i,j}}div \begin{pmatrix} uT\\ \omega T\end{pmatrix} 
		      =\frac{1}{|C_{i,j}|} \int_{\de C_{i,j}}   \mbf{{n}}\cdot\begin{pmatrix} u\\ \omega \end{pmatrix}T,\\
	  &\frac{1}{|C_{i,j}|}\int_{C_{i,j}}div \begin{pmatrix} uq\\ \omega q\end{pmatrix} 
		      =\frac{1}{|C_{i,j}|} \int_{\de C_{i,j}}   \mbf{{n}}\cdot\begin{pmatrix} u\\ \omega \end{pmatrix}q,\\
	  &\frac{1}{|C_{i,j}|}\int_{C_{i,j}}div \begin{pmatrix} u^2\\ \omega u\end{pmatrix} 
		      =\frac{1}{|C_{i,j}|} \int_{\de C_{i,j}}   \mbf{{n}}\cdot\begin{pmatrix} u\\ \omega \end{pmatrix}u,
\end{split}\end{equation}
where $\mbf{ n }$ is the outer normal vector of the cell $C_{i,j}$, $1\leq i \leq N_x$ and $1\leq j \leq N_p$. We then rewrite the second term of \eqref{e:bvp_h} as
\begin{equation}
	  \frac{1}{|C_{i,j}|} \int \int_{C_{i,j}}div 
		      \begin{pmatrix} 
				    u\mathbf{u}\\ \omega \mathbf{u}
		      \end{pmatrix} 
		      \simeq 
		      \frac{1}{|C_{i,j}|}
		      \left(\mathbf{G}_{i,j+\f{1}{2}}
		      -\mathbf{G}_{i,j-\f{1}{2}}+\mathbf{F}_{i+\f{1}{2},j}
		      - \mathbf{F}_{i-\f{1}{2},j}\right),
	  \label{flux}
\end{equation}
where the vertical fluxes $\mathbf{G}_{i,j+\f{1}{2}}$ and $\mathbf{G}_{i,j-\f{1}{2}}$ are respectively the up and down fluxes, and the horizontal fluxes $\mathbf{F}_{i+\f{1}{2},j}$ and $\mathbf{F}_{i-\f{1}{2},j}$ are
respectively the West and East fluxes.

We aim to find $\mbf u_h \in \mc V_h$ with $\omega_h \in \mc W_h$, $\mbf (\Phi_x)_h=(0,0,(\phi_x)_h)$, $(\phi_x)_h\in\mc K_h$, and $\mbf S_h\in (V_h)^3$,
\begin{equation}
\begin{split}
 & \frac{d\mbf u_{i,j}}{dt}=(\mbf R_h(\mbf u_h,\omega_h,t))_{i,j},\enspace 1\leq i \leq Nx,\;1\leq j \leq N_p,\\
&(\mbf R_h(\mbf u_h,\omega_h,t))_{i,j}=-\frac{1}{|C_{i,j}|}
		      \left(\mathbf{G}_{i,j+\f{1}{2}}-\mathbf{G}_{i,j-\f{1}{2}}+\mathbf{F}_{i+\f{1}{2},j}-\mathbf{F}_{i-\f{1}{2},j}\right)\\
&\hspace{1.5in}-(\mbf \Phi_x)_{i,j}+\mbf S_{i,j}.
\label{bvp2}
\end{split}
\end{equation}
Before describing the fluxes, we define the normals vectors. We keep the same direction for all the normal vectors; West to East and Bottom to Top.
Let $\vec{n}_{i,j+\frac{1}{2}}$ be the normal vector for the upper and lower boundaries such that
\begin{equation}
 \vec{n}_{i,j+\frac{1}{2}}=(n^x_{i,j+\frac{1}{2}},n^p_{i,j+\frac{1}{2}}).
\end{equation}
Let $\vec{n}_{i+\frac{1}{2},j}$ be the vector for the East and West boundaries such that
\begin{equation}
 \vec{n}_{i+\frac{1}{2},j}=(1,0).
\end{equation}

The vertical fluxes are defined as follows: for $1\leq i \leq N_x$ and $0\leq j \leq N_p$,
\begin{equation} \label{e:NS_Flux}
\begin{split}
& \mathbf{G}_{i,j+\frac{1}{2}} =|\Gamma_{i,j+\frac{1}{2}}| \vec{n}_{i,j+\frac{1}{2}} . \begin{pmatrix} u_{i,j+\frac{1}{2}} \\\omega_{i,j+\frac{1}{2}} \end{pmatrix} \mathbf{\check{u}}_{i,j+\frac{1}{2}},
\end{split}
\end{equation}
where
\begin{equation}
 \mathbf{\check{u}}_{i,j+\frac{1}{2}} = \begin{cases} \mathbf{u}_{i,j}, \quad \text{if } \quad \vec{n}_{i,j+\frac{1}{2}} . \begin{pmatrix} u_{i,j+\frac{1}{2}} \\ \omega_{i,j+\frac{1}{2}} \end{pmatrix} \geq 0, \\ \mathbf{u}_{i,j+1}, \quad \text{if } \quad \vec{n}_{i,j+\frac{1}{2}} . \begin{pmatrix} u_{i,j+\frac{1}{2}} \\ \omega_{i,j+\frac{1}{2}} \end{pmatrix} < 0. \end{cases}
\end{equation}
We note that the barycenter of the trapezoidal cells are well-aligned in $p$-direction, then we reconstruct $u_{i,j+\frac{1}{2}}$, $u_{i,j-\frac{1}{2}}$, $\omega_{i,j+\frac{1}{2}}$, $\omega_{i,j+\frac{1}{2}}$ using the interpolation method. For instance, we approximate $u_{i,j-\frac{1}{2}}$ and $\omega_{i,j-\frac{1}{2}}$, for $1\leq i \leq N_x$, $1\leq j \leq N_p+1$, by
\begin{equation}\begin{split}
& u_{i,j-\frac{1}{2}} = \frac{u_{i,j} + u_{i,j-1}}{2},\\
&\omega_{i,j-\frac{1}{2}} = \frac{\omega_{i,j} + \omega_{i,j-1}}{2};
\label{interpol}
\end{split}\end{equation}
see e.g. Figure \ref{fig.4.1.1}.
\begin{figure}[H]
\centering
\includegraphics[bb = 0 0 200 200]{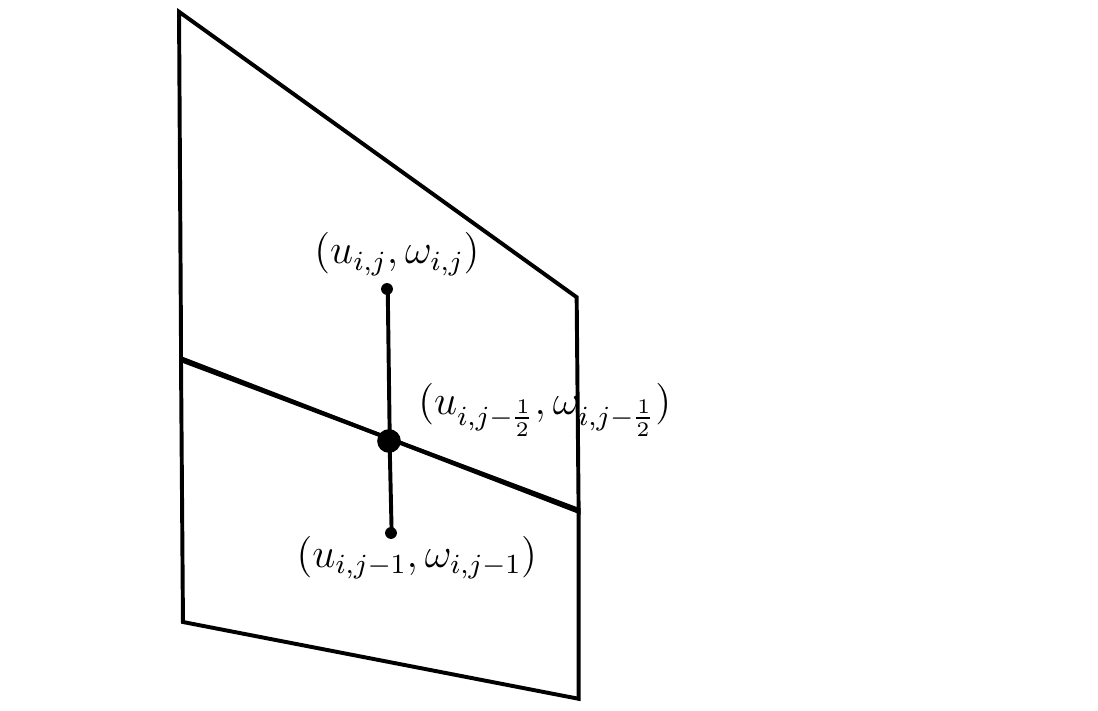}
\caption{Reconstruction of $(u_{i,j-\frac{1}{2}},\omega_{i,j-\frac{1}{2}})$ on the edges where $(u_{i,j-1},\omega_{i,j-1})$ and $(u_{i,j},\omega_{i,j})$ are at the center of each trapezoid cell.}
\label{fig.4.1.1}
\end{figure}
For the horizontal fluxes $\mathbf{F}_{i+\f{1}{2},j}$ and $\mathbf{F}_{i-\f{1}{2},j}$ the normal vectors are $(1,0)$ due to the proposed spatial discretization. (see Figure \ref{mesh1}). Hence, for $0\leq i \leq N_x$ and $1\leq j\leq N_p$, the horizontal fluxes are
\begin{equation}
\begin{split}
	  & \mathbf{F}_{i+\frac{1}{2},j} = |\Gamma_{i+\frac{1}{2},j}| u_{i+\frac{1}{2},j} \mathbf{\check{u}}_{i+\frac{1}{2},j},
\end{split}
\end{equation}
where
\begin{equation}
	  \mathbf{\check{u}}_{i+\frac{1}{2},j} 
		  = 
			\begin{cases} 
			\mathbf{u}_{i,j}, \quad \text{if }  u_{i+\frac{1}{2},j} \geq 0, \\ \mathbf{u}_{i,j+1}, \quad \text{if } u_{i+\frac{1}{2},j} < 0. 
			\end{cases}
\end{equation}
Figure \ref{proj2} shows how we interpolate $u_{i,j}$ and $u_{i+1,j}$ to obtain $u_{i+1/2,j}$.
The expression of $u_{i+\frac{1}{2},j}$, for $0\leq i \leq N_x$, $1\leq j \leq N_p$ reads
\begin{equation}
	  u_{i+\frac{1}{2},j} = ru_{i,j+1} + (1-r)u_{i,j},
\end{equation}
where $r = \frac{p_{i+\frac{1}{2},j} - p_{i,j}}{p_{i+1,j} - p_{i,j}}$.\\
\begin{figure}[h]
\centering
 \includegraphics[scale=0.8]{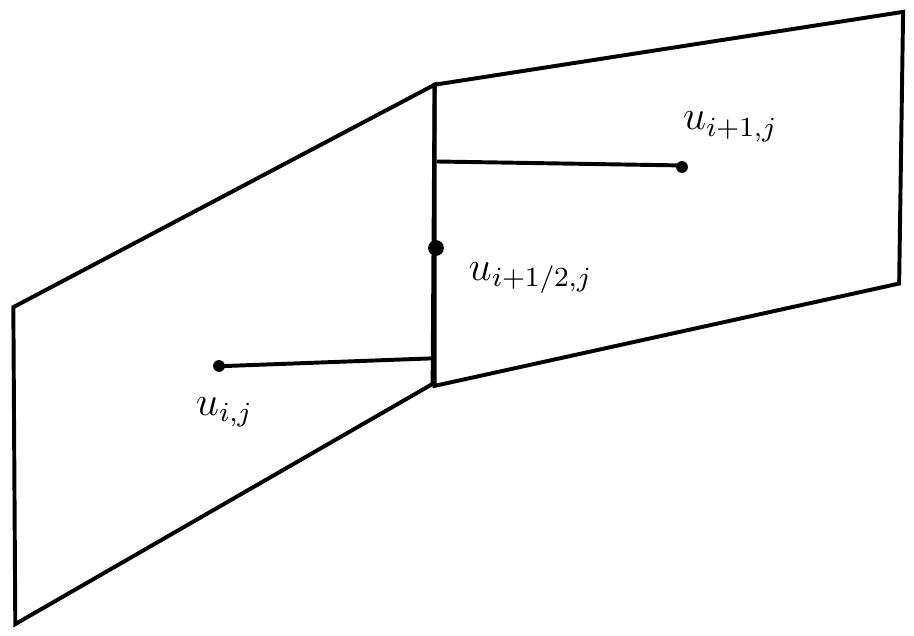}
\caption{Reconstruction of $u_{i+\frac{1}{2},j}$ on the edges where $u_{i,j}$ and $u_{i+1,j}$ are at the center of each quadrilateral cell.}
\label{proj2}
\end{figure}
\subsection{Computation of the projection methods} \label{sec:prom}
In general, the initial condition of $u$ does not follow the compatibility condition \eqref{eq2.2.7}.
Hence, the projection method in Section \ref{ss2.3} plays an important role in our problem.\\
We first set $\alpha_h \simeq \lambda_x$ such that
\begin{equation}
	  \alpha_i=\lambda_x(x_i),\; 1\leq 1 \leq N_x,
\end{equation}
where $x_i$ are the x-coordinates of the barycenter of the cells.

We adopt a forward difference scheme for the derivative in $x$, then \eqref{2.3.4} becomes for $1\leq i \leq N_x-1$
\begin{equation}
	  a_i\alpha_i+b_i\frac{\alpha_{i+1}-\alpha_{i}}{\dx}=c_i
\label{e:lambh}
\end{equation}
where 
\begin{equation}
\begin{split}
	  &a_i = \dfrac{\partial p_B}{\partial x}(x_i),\\
	  &b_i = p_B(x_i) - p_A, \\
	  &c_i = \dfrac{\partial}{\partial x} \int^{p_B}_{p_A} \tilde u dp= 
	  \frac{1}{\dx}\left(\int^{p_B(x_{i+1})}_{p_A} \tilde u
	  dp-\int^{p_B(x_{i})}_{p_A} \tilde u dp\right),
\end{split}
\end{equation}
and $\tilde u$ is the solution $u$ in Section \ref{sec3.2}.
We utilize the mean zero condition in \eqref{e2.14} to impose a boundary condition for $\alpha_{N_x}$:
\begin{equation}
	    \sum^{N_x}_{i=1} \alpha_i = 0 \Longrightarrow 
	    \alpha_{N_x}=-\sum^{N_x-1}_{i=1} \alpha_i.		
\label{e:bd_lamh}	    
\end{equation}
From equations \eqref{e:lambh} and \eqref{e:bd_lamh} we obtain the value of $\alpha_h$ using an LU 
decomposition.

\begin{rem} We consider the Euler method in time, as an example,
and write our projection method as follows:

We define the intermediate steps
 \begin{equation} \label{eq:proj1}
	  \frac{\tilde{u}^{n+1} - u^{n}}{\Delta t}+ div
	  \begin{pmatrix} 
		    (u^n)^2\\ \omega^n u^n
	  \end{pmatrix}
	  +\phi_x^{'n}=0,
\end{equation}
where $\phi'^n_x=(\phi_x)_h$ as in \eqref{e:phi_dis} below, then we find the 
solution $\tilde u^{n+1}$. 
We note that $\tilde u^{n+1}$ does not necessarily satisfy the 
compatibility condition.
We therefore apply the projection method to obtain $u^{n+1}$;
\begin{equation} \label{eq:proj2}
	  \frac{u^{n+1}-\tilde{u}^{n+1}}{\Delta t}+\lambda^{n+1}_x=0.
\end{equation}
Then we easily obtain $\eqref{bvp}_3$ by adding \eqref{eq:proj1} to 
\eqref{eq:proj2}.
	  $$
	  \frac{u^{n+1} - u^{n}}{\Delta t}+ div
	  \begin{pmatrix} 
		    (u^n)^2\\ \omega^n u^n
	  \end{pmatrix}
	  +\phi_x^{'n}+\lambda_x^{n+1}=0,
	  $$
where $\phi_x^{n}=\lambda_x^{n+1}+\phi_x^{'n}$.
\end{rem}

\subsection{Computation of $\omega$} \label{sec:omega}
We look for $\omega$ considering the incompressibility in $\eqref{bvp}_2$ and the given value $u$. We first write a discretized form of $\eqref{bvp}_2$ such that
\begin{equation}
\label{domega}
					\nabla^p_h \omega = -\nabla^x_h u,
\end{equation}
where $\nabla^x_h$ and $\nabla^p_h$ denote the discrete directional derivatives in $x$ and $p$, respectively. Thanks to the proposed spatial discretization, we choose the standard finite difference methods (FDM) for $\nabla^p_h$.
However, for $\nabla^x_h$, it is not accurate to use standard FDM. 
Instead, we utilize finite volume derivatives on $C_{i,j+\f{1}{2}}$ which is defined {in \eqref{diamond}}.
Then, we rewrite \eqref{domega}
\begin{equation}
\omega_{i, j+1} - \omega_{i, j}
        =         -
            (p_{i, j+1} - p_{i, j})
             \nabla_h^x u_h \big|_{K_{i, j+ \frac{1}{2}}},
\text{ }
1\leq i \leq N_x,
\text{ }
1\leq j \leq N_p-1,
\label{domega2}
\end{equation}
with $\omega_{i, 0} = 0$ because $\omega_h \in \mathcal{W}_h$.
We rewrite \eqref{domega2} in a matrix form
\begin{equation}
 \mathcal{A}_h\omega_h = \mathcal{F}_h(u_h),
\label{domega3}
\end{equation}
where $F_h(u_h)$ is the right-hand side and
$\mathcal{A}_h\omega_h$ is the left-hand side of \eqref{domega2}. Since we have
$\omega_{i, 0} = 0$, for $1\leq i \leq N_x$, equation \eqref{domega3} has a
unique solution $\omega \in \mathcal{W}_h$ for a given $u_h \in
\mathcal{V}_h$.

\subsection{Computation of $\phi_x$} \label{sec:phi}
We recall \eqref{e:phi_3} to compute $\eqref{bvp}_3$.
We then project $\phi_x$ onto the space $\mc K_{h}$ in \eqref{e:FV_space3}, and write
\begin{equation}
			(\phi_x)_h = \sum^{N_p+1}_{j=0} \sum^{N_x + 1}_{i=0} (\phi_x)_{i,j} \chi_{C_{i,j}} \in \mc K_{h},
\end{equation}
where $(\phi_x)_{i,j}$ is a step function on $C_{i,j}$ such that $(\phi_x)_{i,j}$ = $(\phi_x)_h  \Big |_{C_{i,j}}$.
We utilize the finite volume derivative on $C_{i,j+\frac{1}{2}}$ to compute $T_x$ as in Section \ref{sec:omega}. Then, \eqref{e:phi_3} becomes
\begin{equation}
			(\phi_x)_{i,j+1} 
			= \sum^j_{j_1=0} (p_{i,j_1 + \frac{1}{2}} - p_{i,j_1 - \frac{1}{2}})
				\frac{-R \nabla^x_h T_h |_{C_{i,j+\frac{1}{2}}}}{p_{i,j}},
\label{e:phi_dis}
\end{equation}
where $1 \leq i \leq N_x$ and $1 \leq j \leq N_p-1$. Considering the boundary condition in \eqref{e:FV_space3}, we complete the computation in \eqref{e:phi_dis}.

\subsection{Time discretization}
\label{sec3.3}
For the time discretization, we use the classical Runge-Kutta 4th-order (RK4) method. Let $t_f>0$ be fixed, denote the time step by $\De t=t_f/N_{t}$ where $N_{t}$ is an integer representing the total number of time iterations; for $n=0,..,N_{t}$ we define $T^{n}$, $q^{n}$, $u^{n}$, $\omega^{n}$ as the approximate values of $T$, $q$, $u$, $\omega$ at time $t_{n}=n\De t$.
We apply the RK4 time discretization using \eqref{bvp2}, \eqref{eq:proj2}, \eqref{domega3}, \eqref{e:phi_dis},  and the boundary conditions defined {in \eqref{e:FV_space1}, \eqref{e:FV_space2}, and \eqref{e:FV_space3}}. We then set\\
\noindent \underline{Step 1}
\begin{equation}
\begin{split}
& \mathbf{k}_{1}^n=\mathbf{R}(\mathbf{u}^{n},\omega^n, t_{n}), \quad \mathbf{\tilde u}^{1,n}=\mathbf{u}^{n}+\De t\mathbf{k}_{1}^n,\\
& u^{1,n}=\tilde u^{1,n}+\De t\lambda_x^{1,n},\quad \ds{\mathcal{A}_h\omega_h ^{1,n}= \mathcal{F}_h(u_h^{1,n})},
\end{split}
\end{equation}

\noindent \underline{Step 2}
\begin{equation}
\begin{split}
& \mathbf{k}_{2}^n=\mathbf{R}(\mathbf{u}^{1,n},\omega^{1,n}, t_{n}+\frac{\De t}{2}), \quad \mathbf{\tilde u}^{2,n}=\mathbf{u}^{n}+\frac{\De t}{2}\mathbf{k}_{2}^n,\\
& u^{2,n}=\tilde u^{2,n}+\frac{\De t}{2}\lambda_x^{2,n}, \quad \ds{\mathcal{A}_h\omega_h ^{2,n}= \mathcal{F}_h(u_h^{2,n})},\\
\end{split}
\end{equation}

\noindent \underline{Step 3}
\begin{equation}
\begin{split}
& \mathbf{k}_{3}^n=\mathbf{R}(\mathbf{u}^{2,n},\omega^{2,n}, t_{n}+\frac{\De t}{2}), \quad \mathbf{\tilde u}^{3,n}=\mathbf{u}^{n}+\frac{\De t}{2}\mathbf{k}_{3}^n,\\
& u^{3,n}=\tilde u^{3,n}+\frac{\De t}{2}\lambda_x^{3,n}, \quad \ds{\mathcal{A}_h\omega_h ^{3,n}= \mathcal{F}_h(u_h^{3,n})},
\end{split}
\end{equation}

\noindent \underline{Step 4}
\begin{equation}
\begin{split}
& \mathbf{k}_{4}^n=\mathbf{R}(\mathbf{u}^{3,n},\omega^{3,n}, t_{n}+\De t), \quad \mathbf{\tilde u}^{n+1}=\mathbf{u}^{n}+\frac{\De t}{6}\left(\mathbf{k}_{1}^n+2\mathbf{k}_{2}^n+2\mathbf{k}_{3}^n+\mathbf{k}_{4}^n\right),\\
& u^{n+1}=\tilde u^{n+1}+\De t\lambda_x^{n+1},\quad \ds{\mathcal{A}_h\omega_h ^{n+1}= \mathcal{F}_h(u_h^{n+1})}.
\end{split}
\end{equation}
\section{Numerical simulations}
\label{sec4}
In this section we carry out numerical experiments.
We first modify \eqref{bvp} by removing $\mbf S$ and adding a source terms so 
that we see the effectiveness of the proposed scheme {from Section 
\ref{sec3.3}}.
{In Section \ref{sec4.1}, we test our scheme and {estimate its rate of 
convergence, numerically}.}
In Section \ref{sec4.2} and \ref{sec4.3}, {we perform physically 
plausible computations by solving} the full equations \eqref{bvp} supplemented 
with the proper boundary conditions in \eqref{bc}.
\subsection{Analytic case}\label{sec4.1}
In this section we use the following system of equations:
\begin{equation}
\begin{cases}
\label{eq6.1.1}
	  \displaystyle{\dfrac{\de \mbf u }{\de t}+\nabla_{x,p}(u\mbf u,\omega \mbf u)=\mbf B(\mbf u, \omega,t ),}\\
	  \ds{\omega=-\int_{p_A}^p\dfrac{\de u}{\de x}dp,}
\end{cases}
\end{equation}
where $\mbf B=(B_T,B_q,B_u)$ corresponds to the source terms derived by the analytical solution defined below.
We note that $\mbf B$ is different from $\mbf S$ in \eqref{e:notation}.
We set the domain as $[0,L] \times [p_A,p_B(x)]$ where $0 = 0$, $L=50,000$, $p_A = 100$ and $$p_B(x) = 1000 -200\exp\left(-\frac{(x-25000)^2}{3000^2}\right).$$
The function $p_B$ satisfies $\eqref{bc}_1$ approximately. Indeed, we can easily calculate
$$\frac{\de p_B}{\de x}\simeq10^{-31} \text{  at  } x=\{0,L\},$$ and the quantity is negligible compared to the other numerical errors.\\
We add the boundary conditions \eqref{bc} and the divergence free condition to $\eqref{eq6.1.1}_4$, we then choose:
\begin{equation}
\begin{split}
	  T_{EX}(x,p,t) =& - \frac{p}{R} \frac{\partial \phi}{\partial p}, \quad q_{EX}(x,p,t) = 0,\\
	  u_{EX}(x,p,t) =& -\frac{\partial \xi}{\partial p},\quad w_{EX}(x,p,t) = \frac{\partial \xi}{\partial x},
	\end{split}
 \label{exsol}
\end{equation}
and
\begin{equation}
	  g_T(p)=- \frac{p}{R} \frac{\partial \phi}{\partial p}\Big|_{x=0,t=0},\enspace 
	  g_q(p)=0, \enspace g_u(p)=-\frac{\partial \xi}{\partial p}\Big|_{x=0,t=0},
\end{equation}

where
\begin{equation}
\begin{split}
	  \xi(x,p,t) =& \left ( \frac{p-pA}{100} \right)^3 \left (\frac{p-pB(x)}{100} \right )^3 \cdot(\cos(2 \pi t) + 20)\cdot\frac{x^3(x-L)^3}{L^6}, \\
	  \phi(x,p,t) =& \left[ \Big (\frac{p-pB(x)}{450} \Big )^3+ \left\{(-R(T_0 - \Delta T)\log(p)- R \frac{\Delta T}{p_0} p\right.\right. \\
	      & +R (T_0-\Delta T) \log(p_0) + R\Delta T)\Big\}/g\Big]\cdot \cos(2 \pi t)\cdot\frac{x(x-L)^2}{L^3}.
\end{split}
\end{equation}
Using these analytic functions, we find the rate of convergences for the 
proposed scheme.
In the simulations, we set $\Delta t = 10^{-2}$, and $[Nx,Np]$ = $[100,100]$, 
$[150,150]$,  $[200,200]$, $[250,250]$, and $[300,300]$ to check the 
convergence.\\
Table \ref{tb1} and Figure \ref{f1} show the relative $L^2$ errors at $t_f=\De t \times k$, where $k=100$, for different spatial discretizations.
Here, we define the relative $L^2$ error for e.g. $T$ by
\begin{equation}
	\|T\|_{Error}:=  \sqrt{\frac{{ \sum_{i,j} |C_{i,j}| \left[T_{EX}(x_{ij},p_{ij},t_k)-T_{NUM}(x_{ij},p_{ij},t_k)\right]^2}}{{ \sum_{i,j} |C_{i,j}|  T^2_{EX}(x_{ij},p_{ij},t_k)}}},
	\label{L2_er}
\end{equation}
where $T_{EX}$ is the exact solution in \eqref{exsol} and $T_{NUM}$ is the numerical solution of \eqref{bvp2} in \eqref{eq6.1.1}.
Also, we denote $|C_{i,j}|$ is the area of the $(i,j)$-th cell.
The relative $L^2$ errors for $u$ and $\omega$ are obtained in the same way.
In Table \ref{tb1}, we observe the rate of convergence of the numerical solutions $T$, $u$ and $\omega$.
Figure \ref{f1} shows a first order convergence of the numerical solutions for our scheme, {with $t_f=100\Delta t$. We observe the same rate of convergence for greater $t_f$.} \\
\begin{table}[H]
\caption{Relative $L^2$ errors for \eqref{eq6.1.1} with \eqref{bc} at $t_f=100\De t$ where $\De t = 10^{-2}$} 
\centering 
\begin{tabular}{c c |c c c} 
 \\[0.1ex] \hline\hline 
$N_x$ & $N_p$ & $\|T\|_{Error}$ &$\|u\|_{Error}$&$\|\omega \|_{Error}$\\ [0.5ex] 
\hline 
 100 & 100  & 7.209e-07 &  1.023e-04   &  1.466e-02 \\ 
150 & 150 & 4.002e-07  &  6.722e-05  &  6.615e-03    \\
200 & 200 & 2.631-07 &  5.014e-05   & 3.764e-03    \\
250 & 250 & 1.904e-07 &  3.997e-05   & 2.435e-03    \\
300 & 300 & 1.466e-07 & 3.325e-05 &   1.708e-03 \\[1ex]
\hline 
\end{tabular}
\label{tb1} 
\end{table}

\begin{figure}[H]
\centering
 \includegraphics[scale=.7]{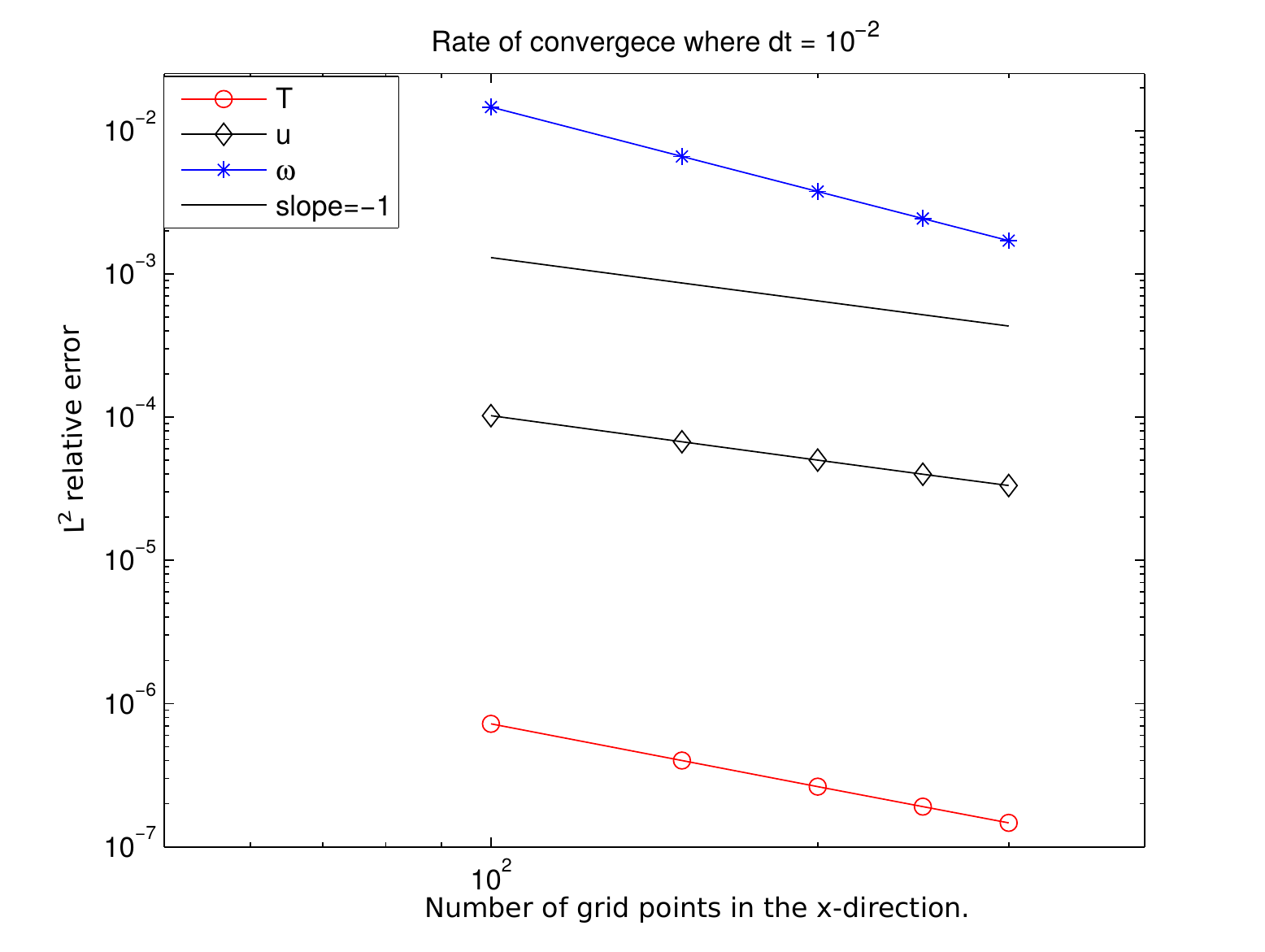}
\caption{The convergence of the relative $L^2$ errors in log-log scale based on Table \ref{tb1}. The slopes of the log-log plots for $\|T\|_{Error}$, $\|u\|_{Error}$ and $\|\omega\|_{Error}$ (defined in \eqref{L2_er}) are 1.44, 1.02, and 1.95, respectively.} 
\label{f1}
\end{figure}
\noindent

\subsection{Physical case: deterministic simulations}
\label{sec4.2}
In this subsection we solve \eqref{bvp} with the physical boundary conditions in \eqref{bc}.
To perform numerically stable computations, we consider an averages in space.
For instance, we update the values of $T_h^m$ for some time step $m$:
\begin{equation}
	    T_{i,j}^{m}=\frac{ T_{i,j}^{m}+ 
T_{i-1,j}^{m}}{2},\mbox{ for } 1 \leq i \leq N_x.
\end{equation}
In our simulation we average $T$ for every 18 time step, i.e. $m=18n$. For 
$u$, $\omega$, and $\phi$, 
we average the cells in the same way but we take $m=n$.
\\~~\\
\\
\noindent\textbf{The initial conditions.}\\
\\
We recall that the temperature can be written as $T(x,p,t)=\bar{T}(p)+T'(x,p,t)$ and we take $T'(x,p,t=0)=0$; therefore
\begin{equation}
\label{T02}
T(x,p,0)=\bar{T}(p)=T_0-\left(1-\frac{p}{p_0}\right)\Delta T,
\end{equation}
where $T_0=300 K$ and $\Delta T=50$.
The initial value of the humidity $q$ is
\begin{equation}
	  	  q(x,p,t=0) = q_s - 0.0052,
\label{q0}
\end{equation}
where $q_s$ is the saturation defined in \eqref{qs}.
Figure \ref{f:q_init} shows the shapes of the initial condition for $q$ and the 
saturation $q_s$ at a certain height (around 200m away from the earth). Note that we choose a slightly under saturated initial condition for $q$ to see how the mountains produce saturation and rain, that is,
$$
	  q(x,p,t=0) < q_s.
$$
\begin{figure}[h]
    \centering
    \includegraphics[scale=0.45]{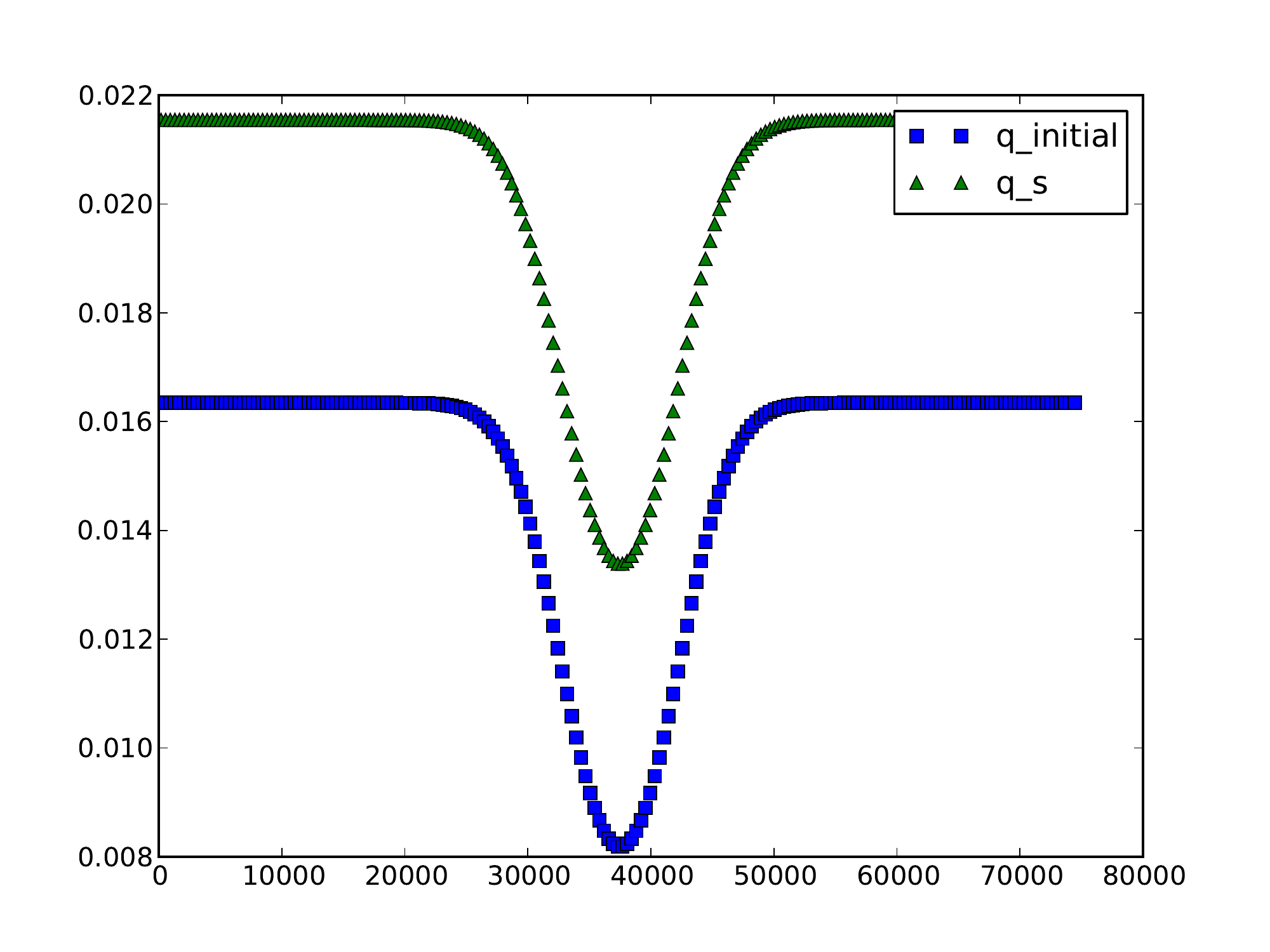}
    \caption{$q$ and $q_s$ at $t=0$ at a certain height (around 200m away from the earth).}
\label{f:q_init}
\end{figure}
To describe the initial condition for the velocity $u$, we first 
introduce the intermediate value $\tilde{u}$ which the velocity before it is 
projected.
We then write 
$$\tilde{u}(x,p,t)=\bar{u}+u'(x,p,t),$$
where
\begin{equation}
	  \bar{u}=7.5m/s,~~u'(t=0)=2\cos\left( \frac{p\pi}{p_0}\right)\cos\left(\frac{2n\pi x}{L}\right).
\end{equation}
This gives the initial condition for $\tilde u$:
\begin{equation}
\label{u0}
	  \tilde{u}(x,p,0) = 7.5 + 2\cos\left( \frac{p\pi}{p_0}\right)\cos\left(\frac{2n\pi x}{L}\right).
\end{equation}
Note that $\tilde u$ does not follow the compatibility condition in 
\eqref{eq2.2.7c}. Hence, after applying the projection method {described 
in \eqref{eq:proj1} and \eqref{eq:proj2},}
we obtain the  initial condition for $u$, that is, $u = \mc F_h\left( \tilde u \right)$.
{Figure \ref{proj_phy} reports on the striking numerical advantage of using such a method for a simulation of  $u$ in good agreement with the natural constraints associated with the problem at hand such as the compatibility condition 
defined in \eqref{eq2.2.7c} emphasizing the constant horizontal profile that a vertical integration of $u$ must satisfy. As one can observe on Fig.~\ref{proj_phy}, when the projection method
is applied,   the deviations from such a constant horizontal 
profile  are reduced by a factor $10^{4}$, {\mk reducing in other words, the error of the simulated $u$ in satisfying \eqref{eq2.2.7c} by the same factor.}}

For the boundary condition at $x=0$, we use
\begin{equation}
 \begin{split}
& g_T(p)=\bar T(p),\\
& g_q(p)=q_s(\bar T,p),\\
& g_u(p)=\mathcal F_h\left(\bar u + 2\cos\left(\frac{p\pi}{p_0}\right)\right).
 \end{split}
\end{equation}\\
In the following simulation, we choose the parameters:
\begin{equation}
\begin{split}
 & p_A = 250, \quad [0,L] = [0,75000],\\
 & [Nx,Np]=[200,200], \quad t \in [0,20000], \quad \De t= 0.5,\\
 & p_B(x) = 1000 -250\exp\left(-\frac{(x-37500)^2}{6000^2}\right)
 \end{split}
\end{equation}

\begin{figure}[H]
\centering
	\includegraphics[scale=.6]{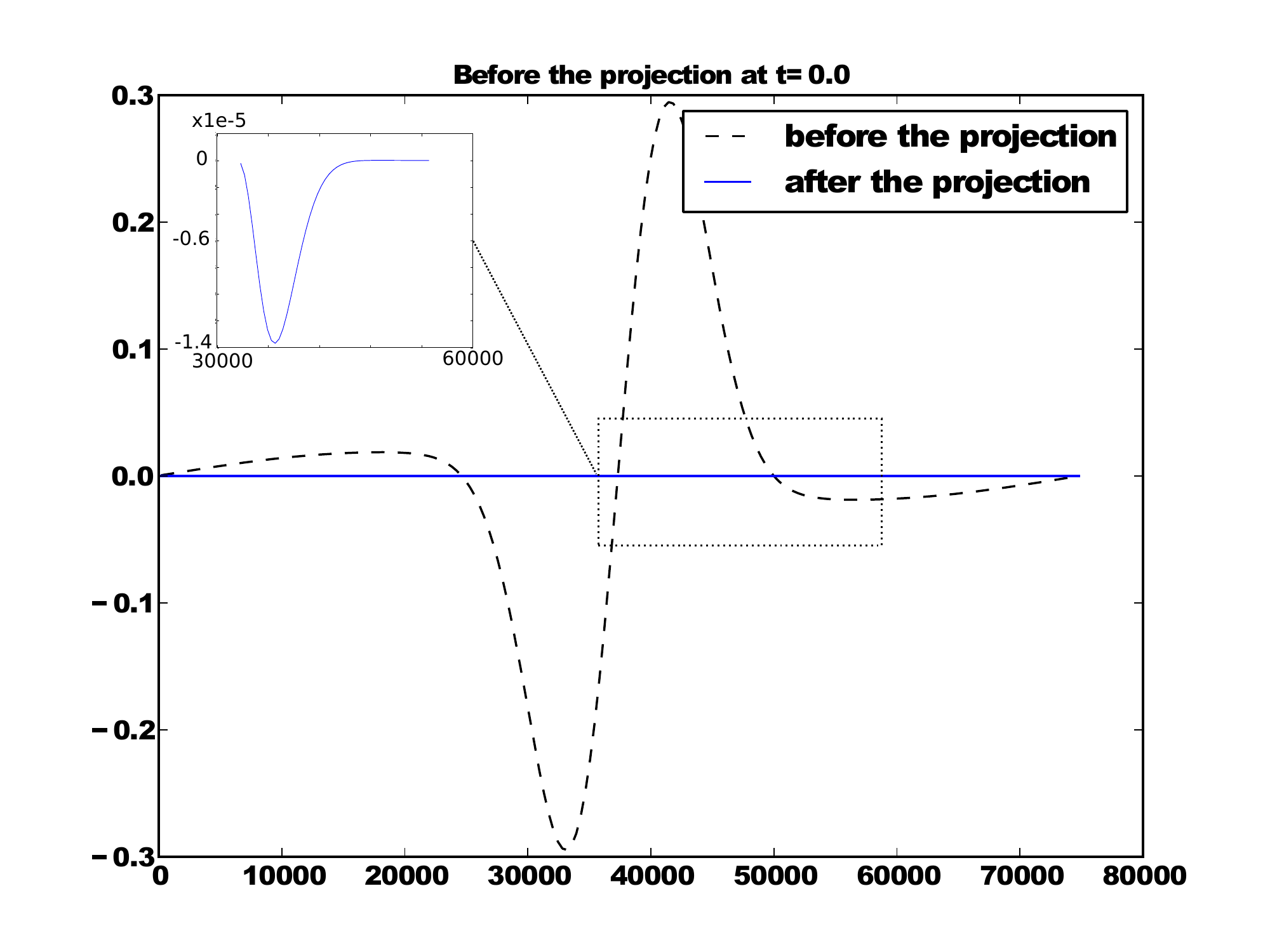}
\caption{ Computations of the quantity $Q=\frac{\partial }{\partial 
x}\int_{p_A}^{p_B}u \text{ }dp$, before and after the 
projection.}
\label{proj_phy}
\end{figure}
Figure \ref{f4} shows {\mk snapshots} of $T$ (temperature) and $q$ (humidity) at different times. Since $u$ is positive, the flow moves from West to East.
On the upstream side of the mountain, the temperature is lower and the humidity 
is higher, whereas on the downstream side of the mountain, the temperature is 
higher and the humidity is lower for sufficiently large $t$. These results are 
coherent with the physical context.
Note that in Figure \ref{f4} we magnify the value of $T$ and $q$ near the 
ground to see in detail the behavior of $T$ and $q$.
Figure \ref{f4_uw} shows behaviors of $u$ (horizontal velocity) and $\omega$ (vertical velocity).
In Figure \ref{f:profile_T_q}, we present the 1D curve for $T$ and $q$ at $t=15000$ along the mountain, i.e. along the dotted line in Figure \ref{f:slice_1D}. We observe that the figures show asymmetries: it is more humid on the left, and warmer on the right.
There are consistent with the fact that it rains on the left side of the mountain, where the wind comes from.
Figure \ref{f:L2_energy} {\mk shows the time-evolution of the (spatial) $L^2$-norm
of $T$, $q$, $u$, and $\omega$. We observe that the numerical solutions {\mk reach a steady state for sufficiently large $t$ (e.g. $t > 15000$), while the time-evolution the $L^2$-norm of $T$, $q$, and $u$ exhibit a transient growth that suggests the presence of nonnormal modes which in the present context can be explained as resulting from the topography which breaks the symmetry of the (physical) domain leading typically to non-orthogonal modes associated with the linearized operator. 
Under these circumstances, disturbances can develop in the system that is favorably
configured to undergo rapid transient growth, even in the absence of any growing modes.
Because the modes are non-orthogonal, they have a non-zero projection on one another,
so it is possible to superpose them to produce disturbances that initially grow rapidly.  Such a growth can be further amplified by nonlinear effects or by noise such as documented in the literature; see {\it e.g.} \cite{penland1995optimal} and the supporting information of \cite{chekroun2011predicting} for an analogous situation in a simple model. We turn now to the investigation of such a phenomena in the next section by including some random small-scale disturbances in the model formulation that as we will see cause the appearance of propagating (modulated) waves in response to such a noise-forcing. We refer to \cite{farrell1988optimal},\cite{farrell1996generalized}, \cite{moore1997singular}, \cite{reddy1993pseudospectra}, \cite{trefethen1993hydrodynamic}, and \cite{trefethen2005spectra} for manifestations of nonnormal modes in various settings borrowed from hydrodynamic or geophysical fluid dynamics.}
}
\begin{figure}[H]
\centering
	  \includegraphics[scale=.35]{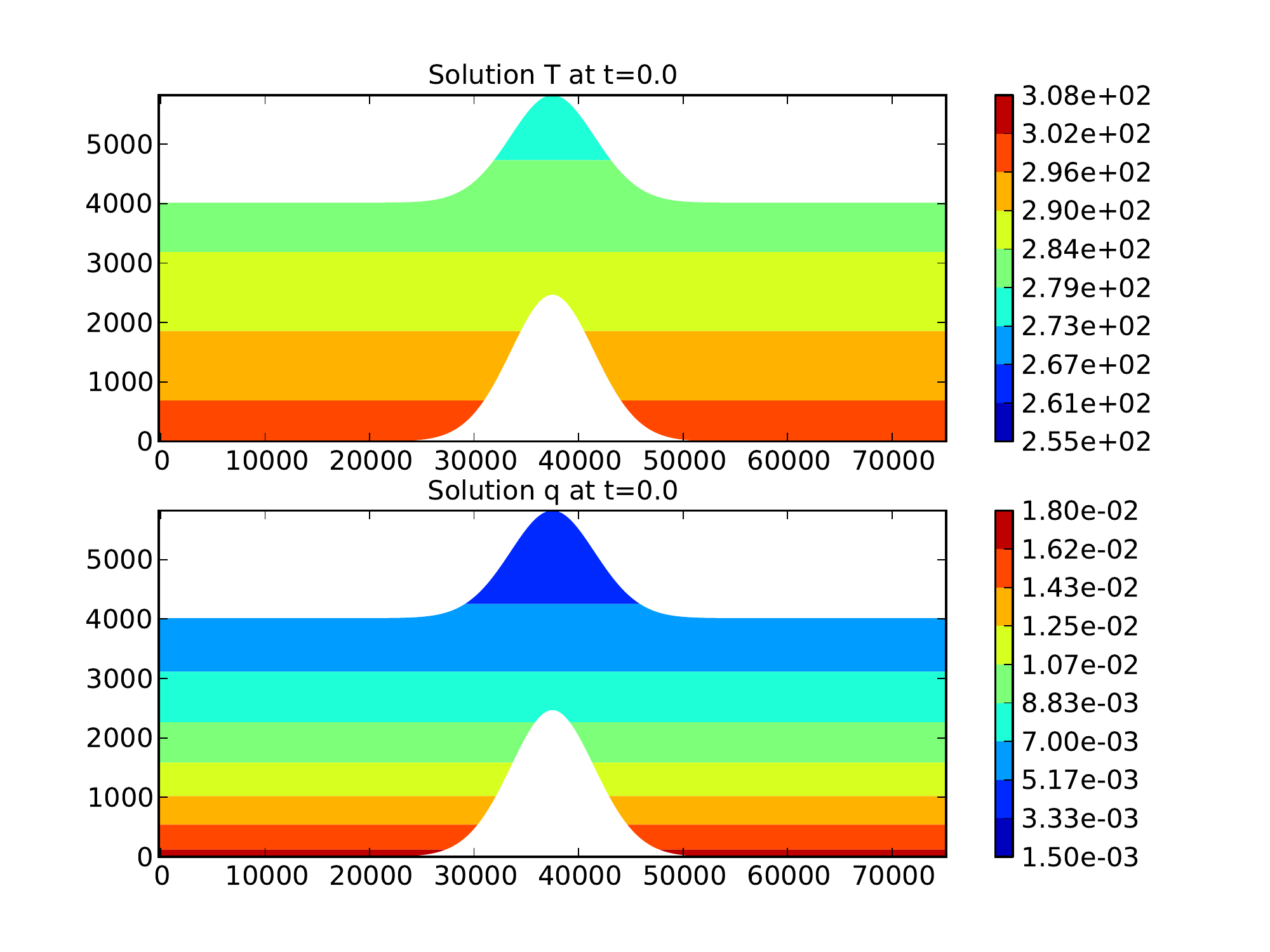}
	  \includegraphics[scale=.35]{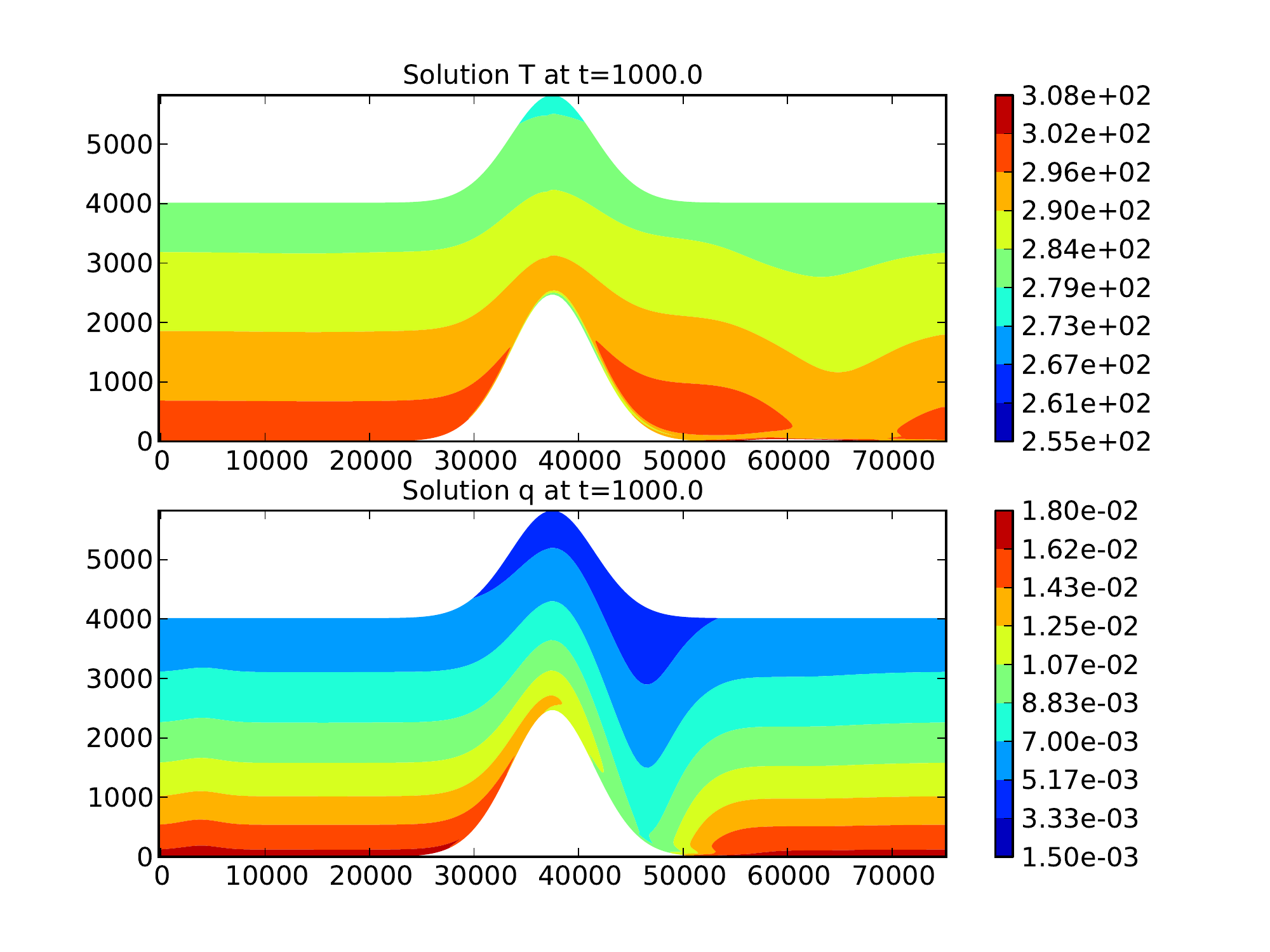}
	  \includegraphics[scale=.35]{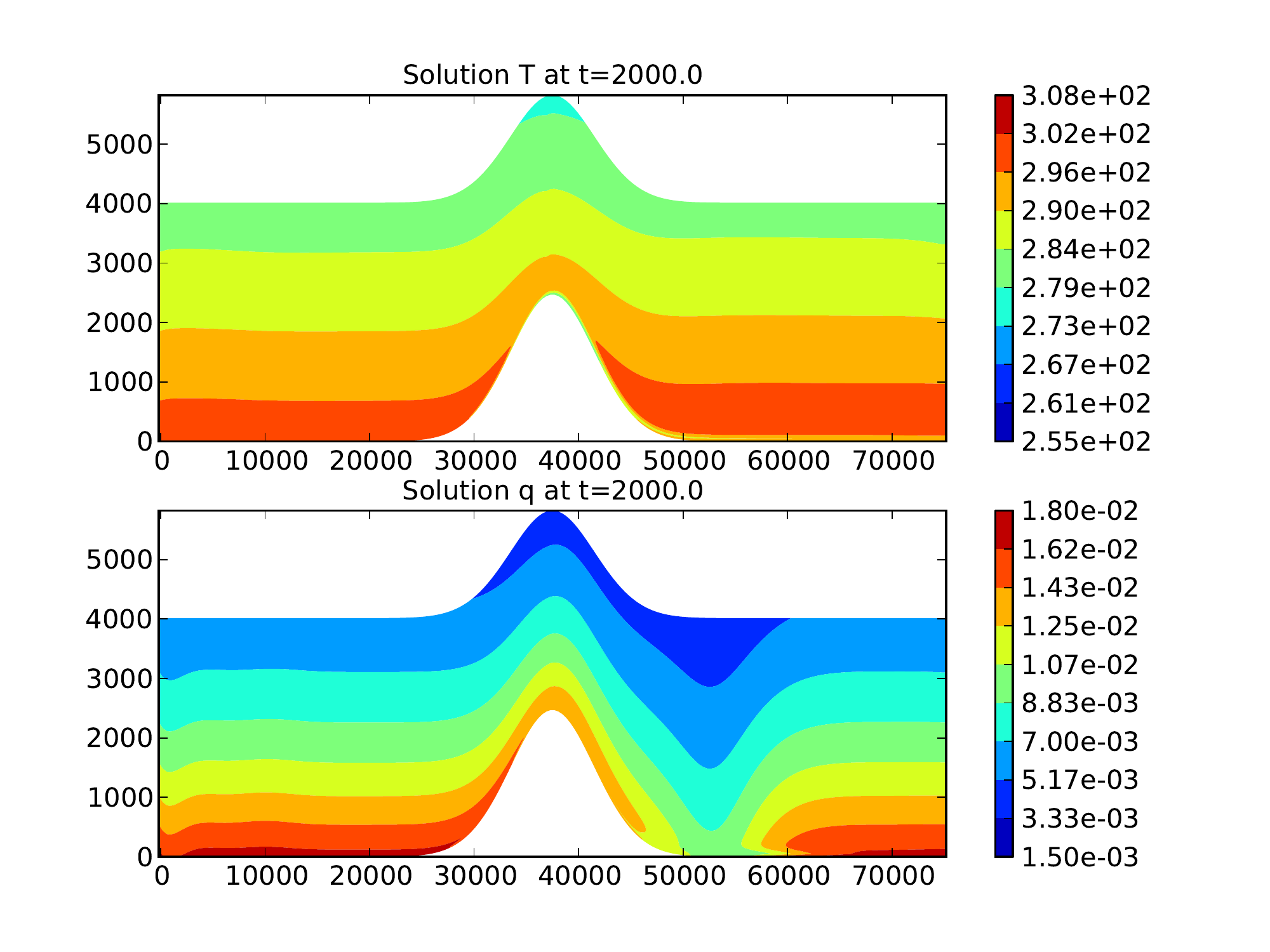}
	  \includegraphics[scale=.35]{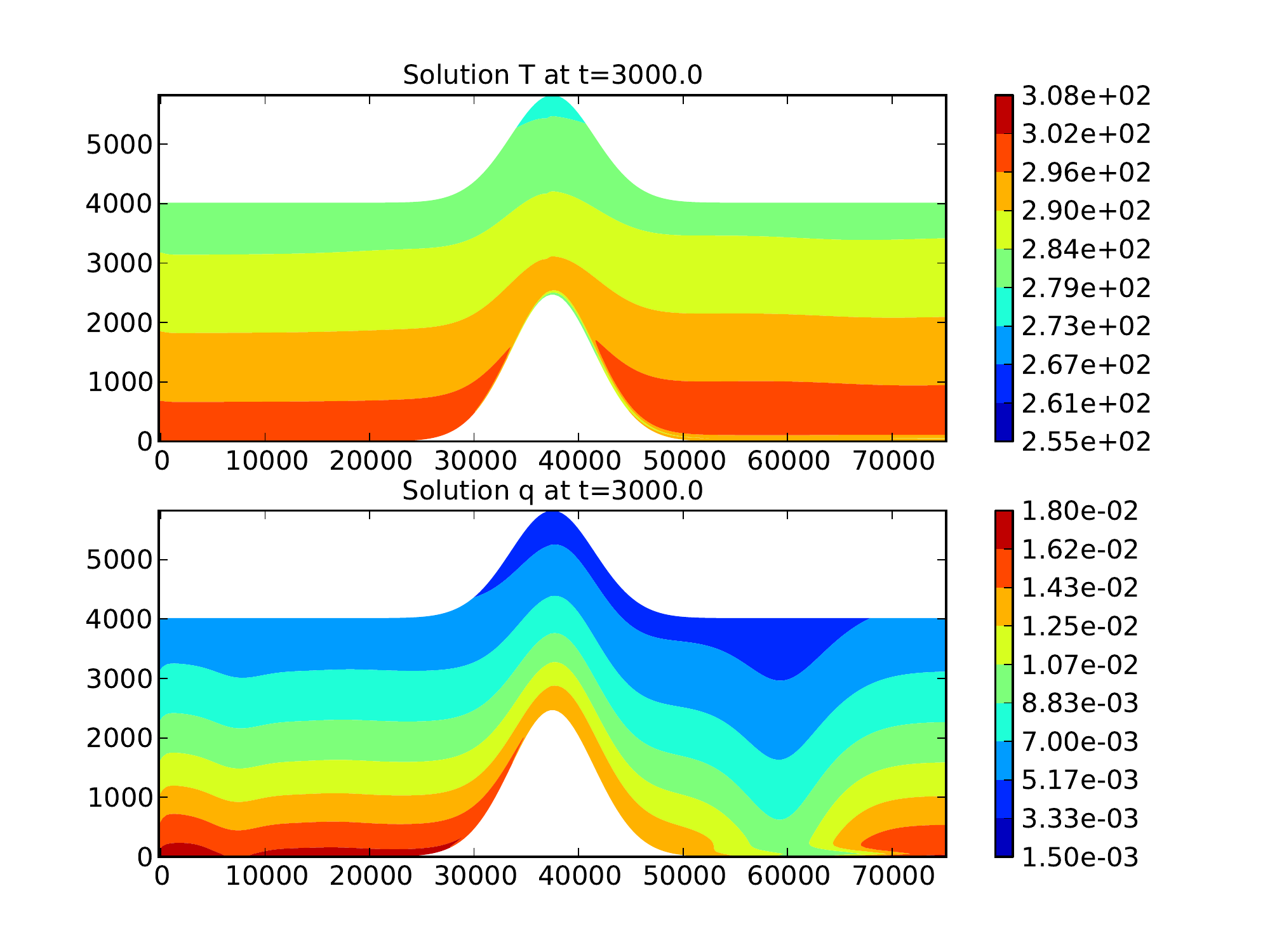}
	  \includegraphics[scale=.35]{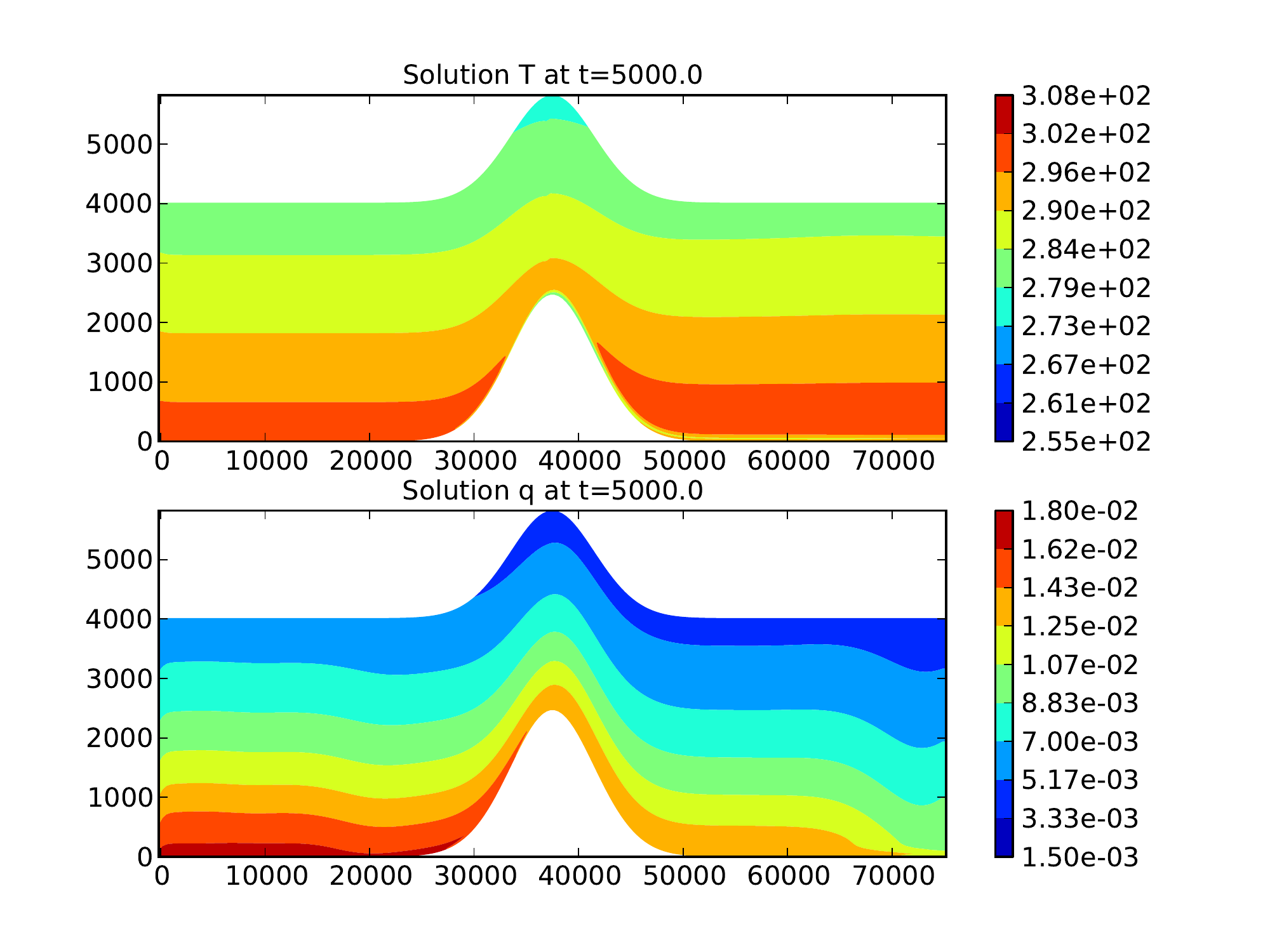}
	  \includegraphics[scale=.35]{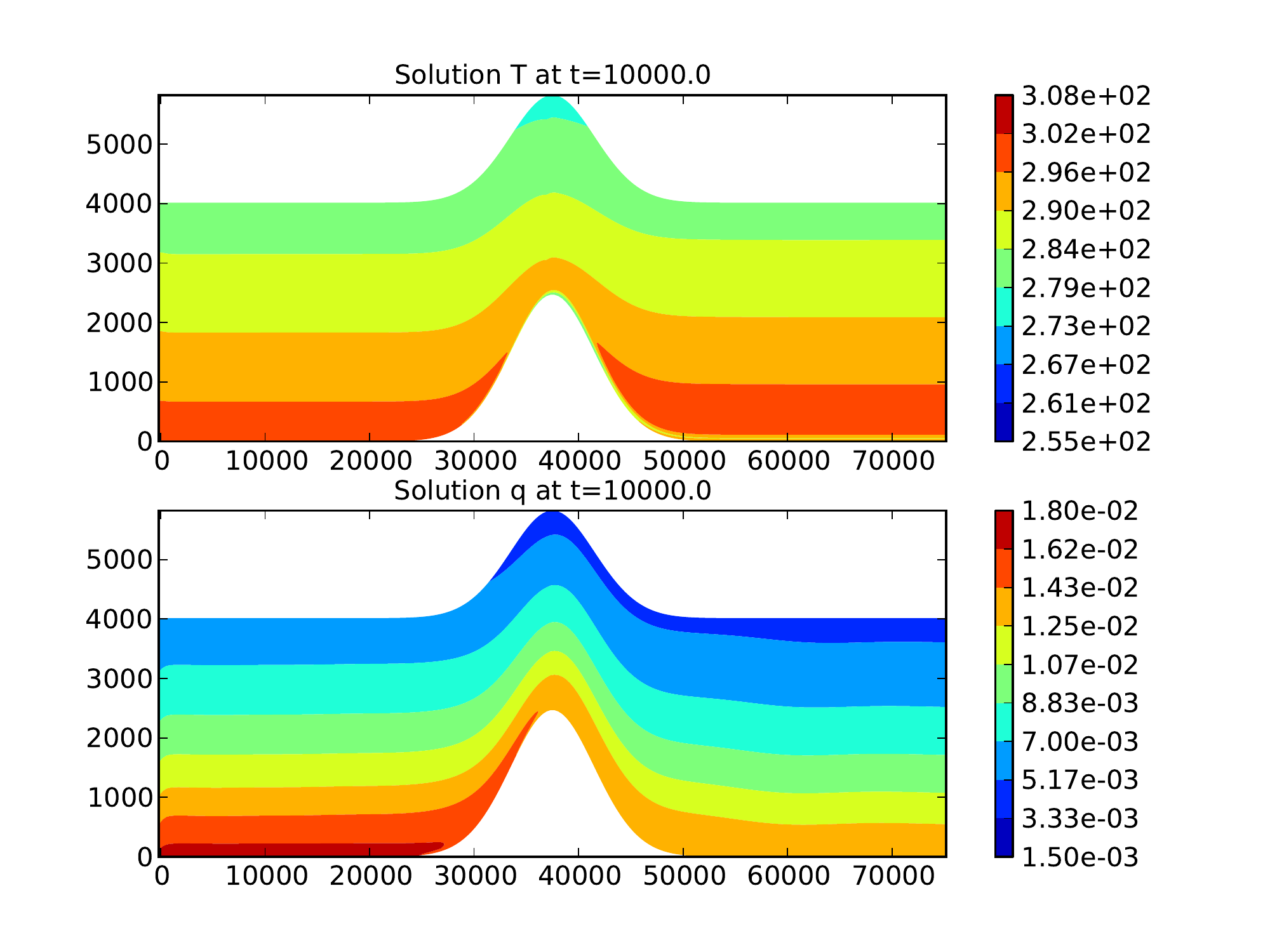}
\caption{Solutions of \eqref{bvp} and \eqref{bc} for $T$ and $q$ at $t=$0, 1000, 2000, 3000, 5000, 10000.}
\label{f4}
\end{figure}

\begin{figure}[H]
\centering
	  \includegraphics[scale=.35]{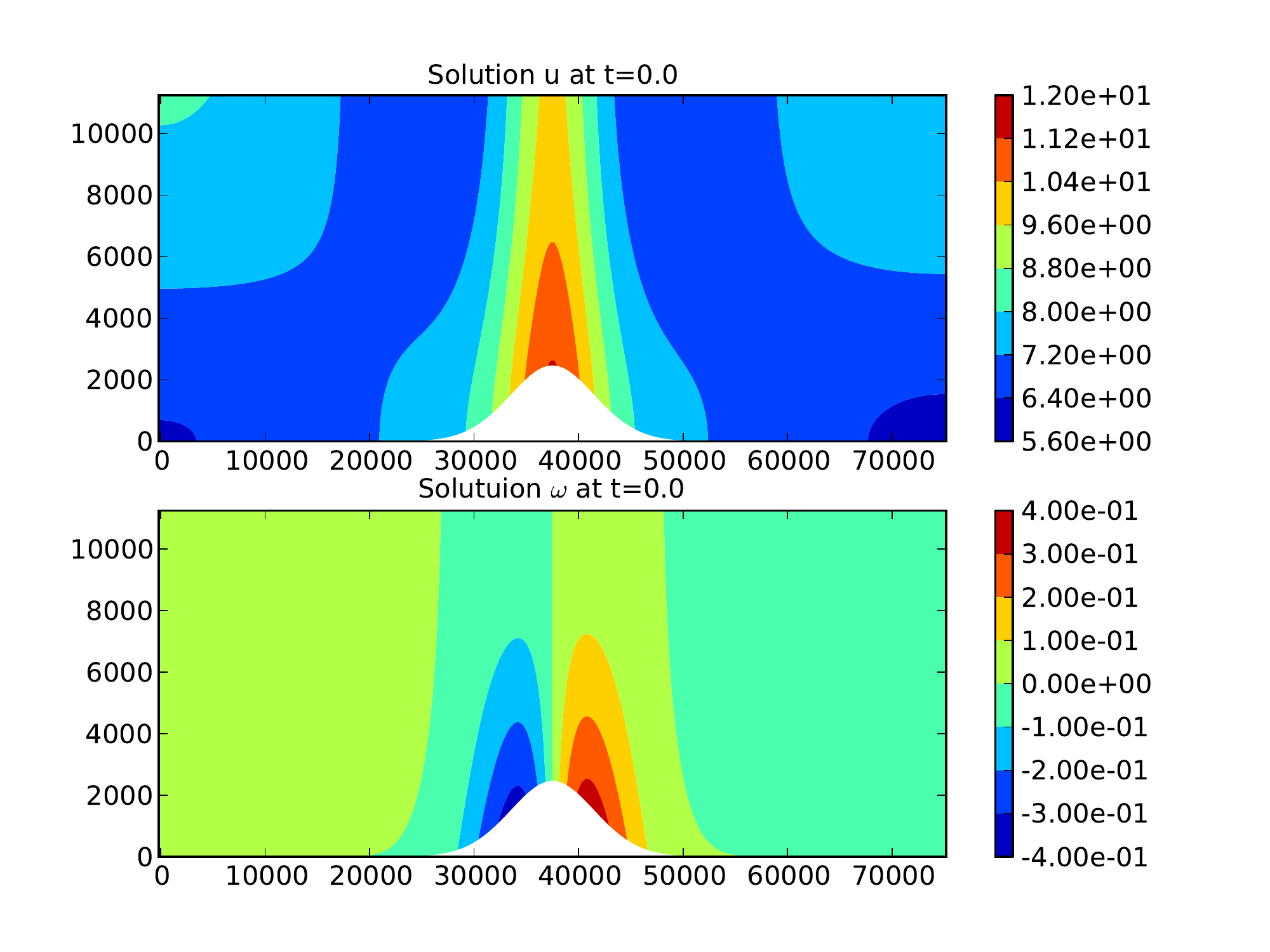}
	  \includegraphics[scale=.35]{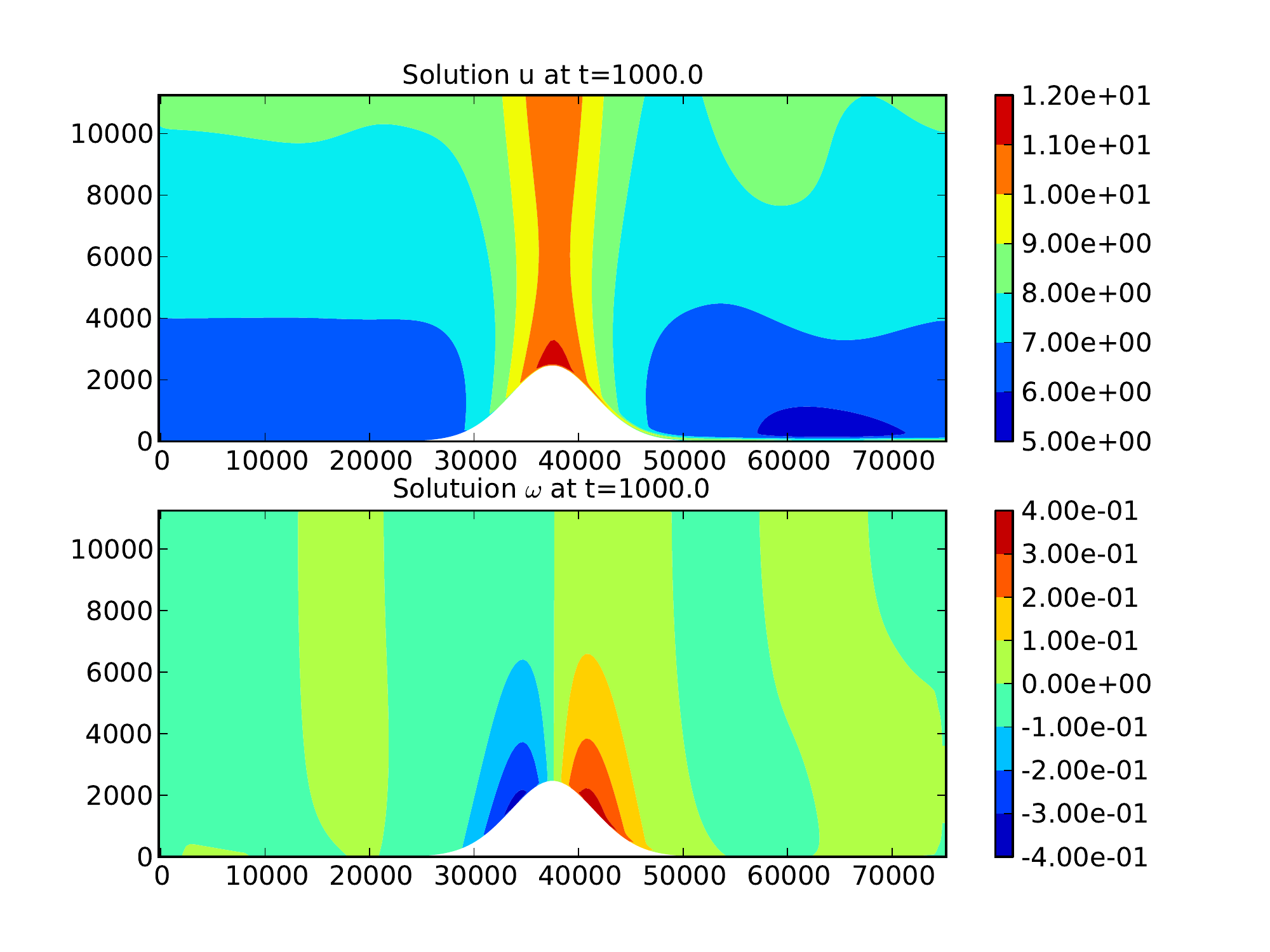}
	  \includegraphics[scale=.35]{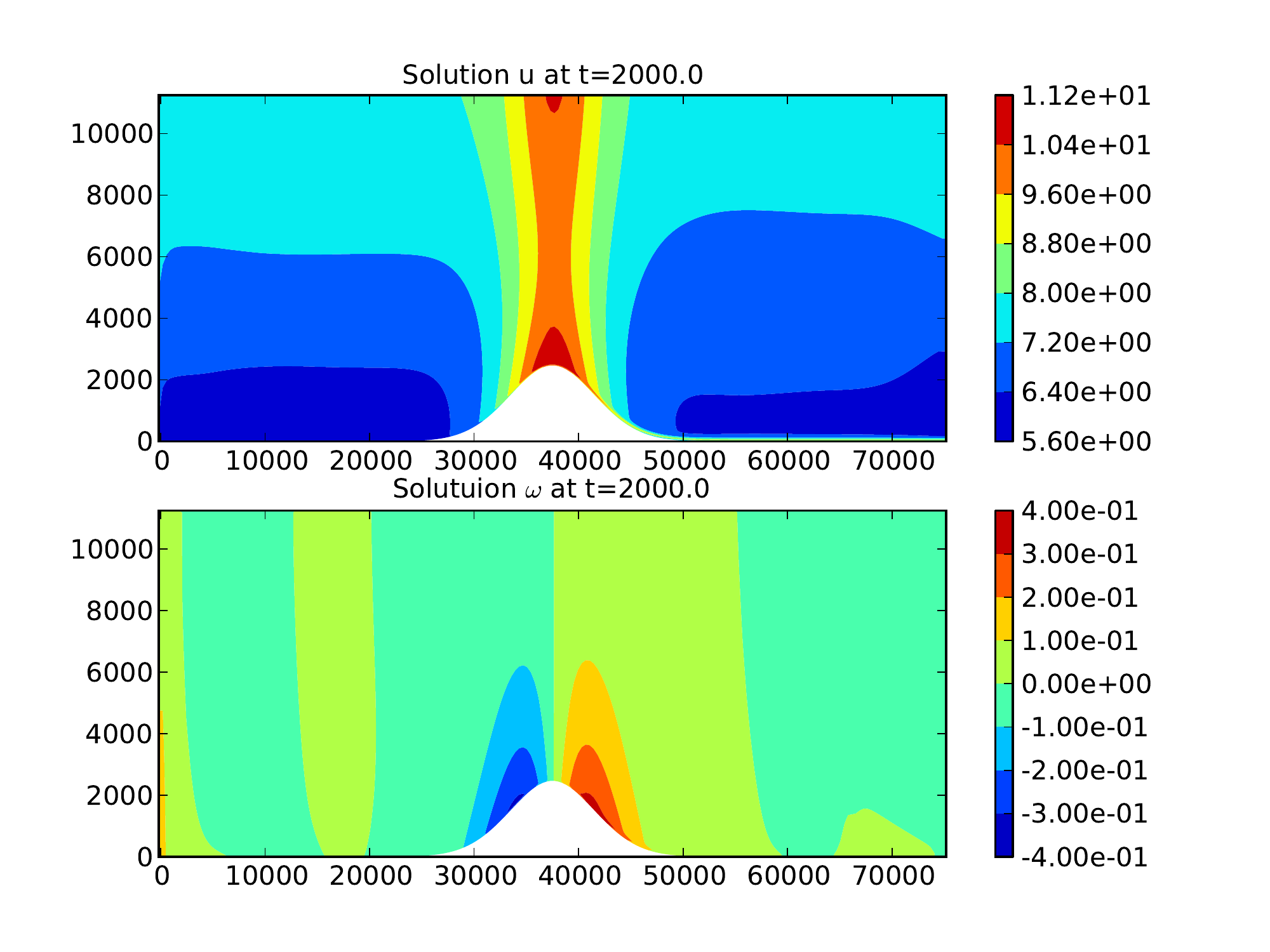}
	  \includegraphics[scale=.35]{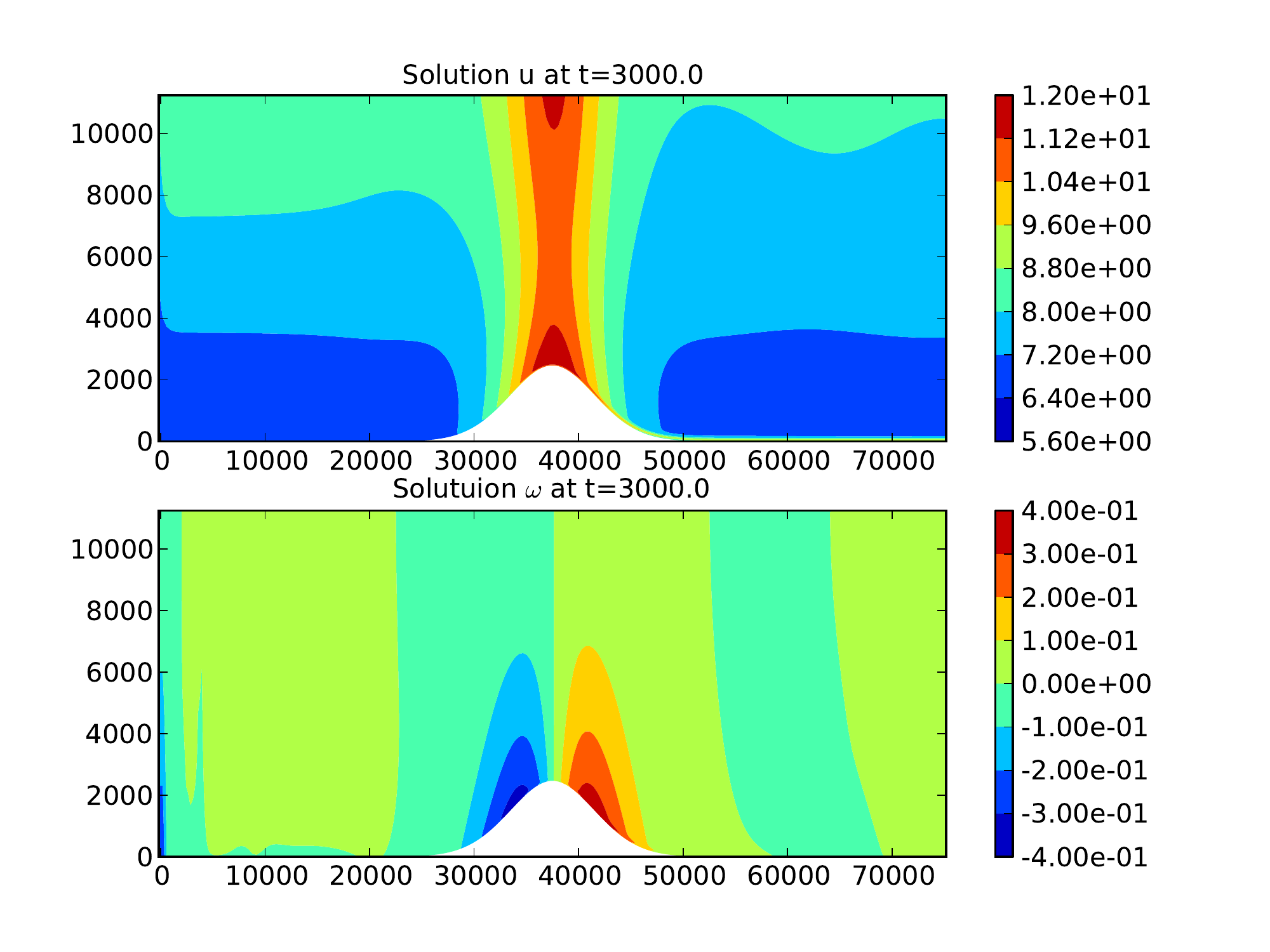}
	  \includegraphics[scale=.35]{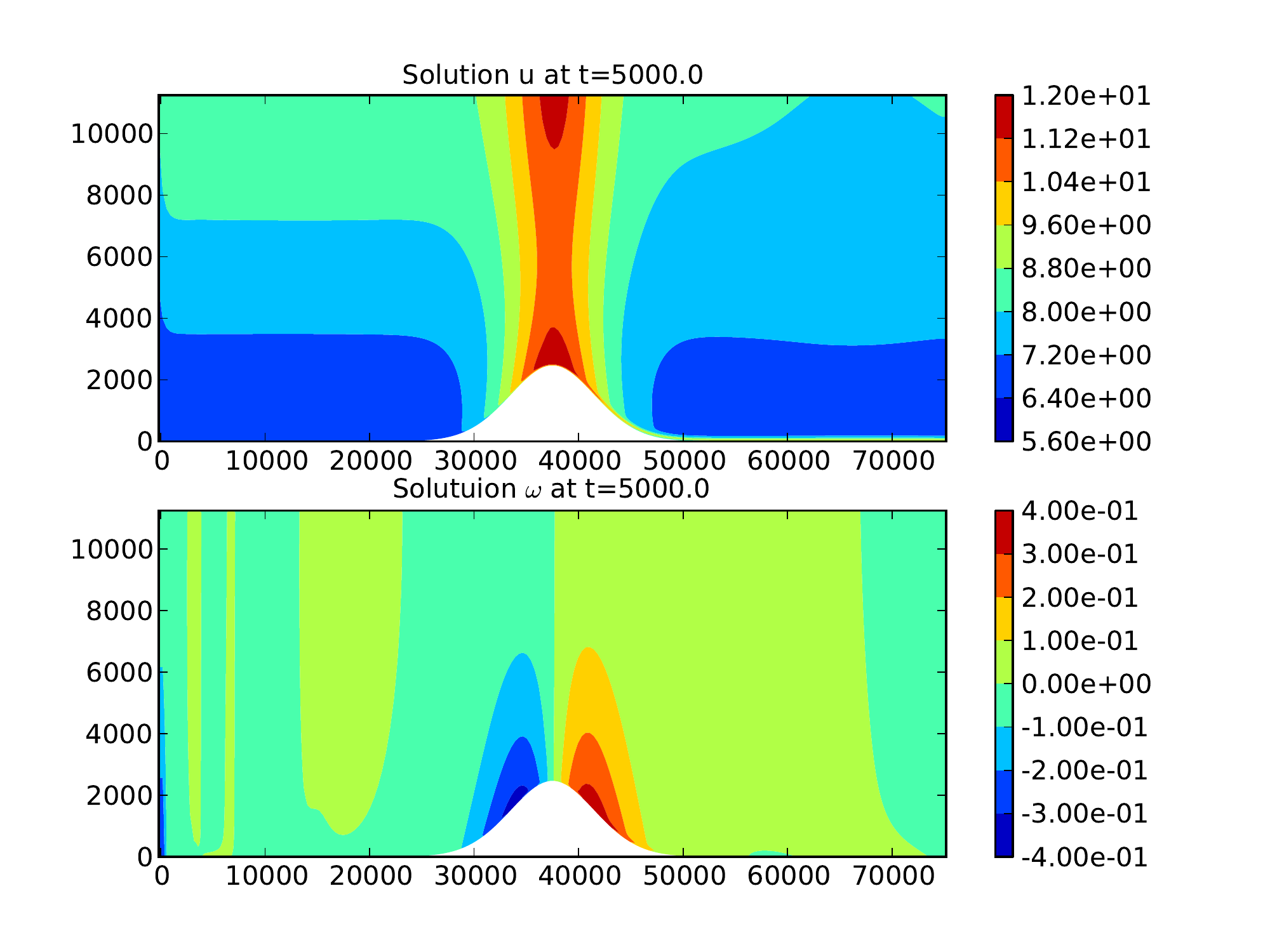}
	  \includegraphics[scale=.35]{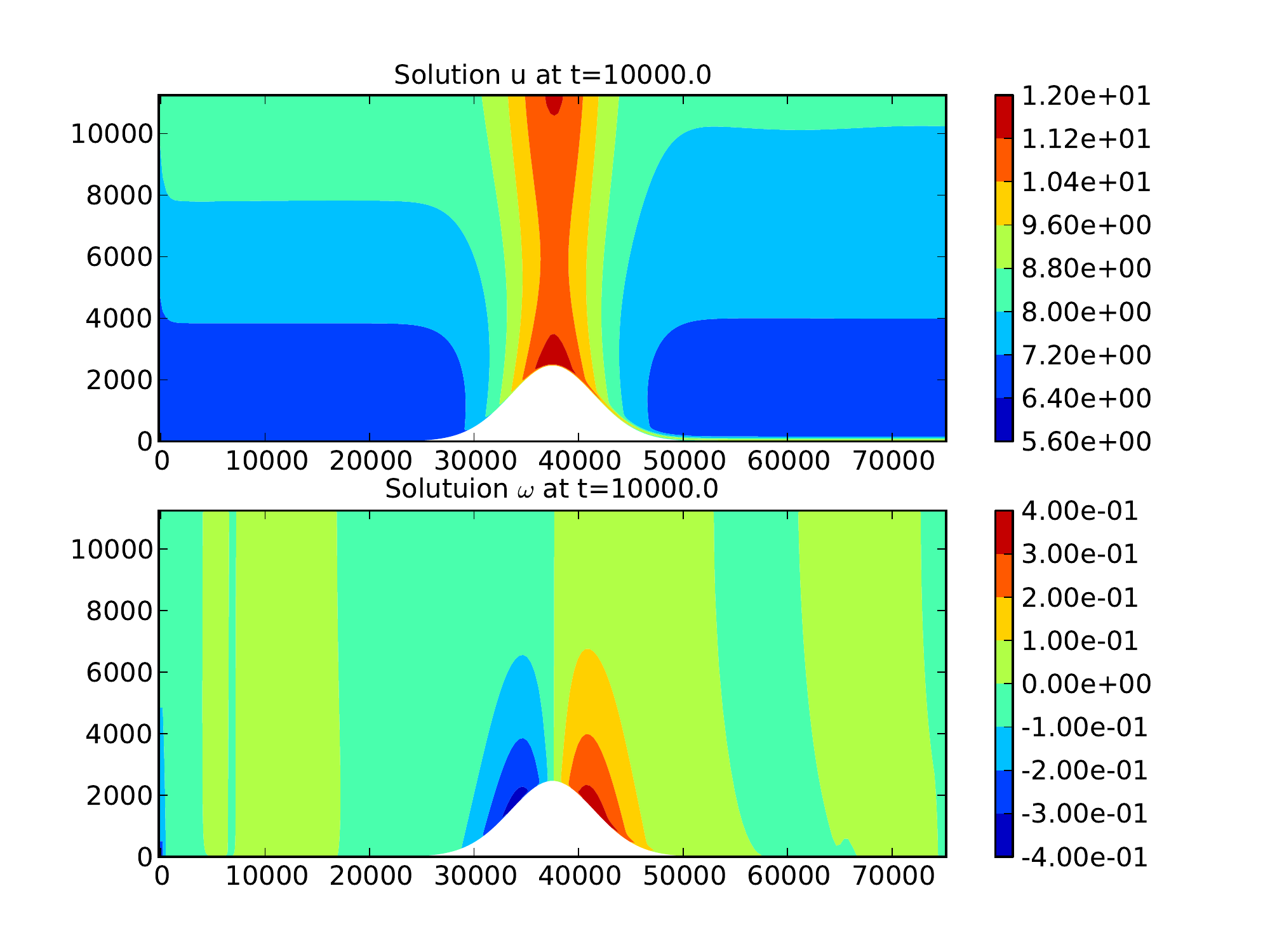}
\caption{Solutions of \eqref{bvp} and \eqref{bc} for $u$ and $\omega$ at $t=$0, 1000, 2000, 3000, 5000, 10000.}
\label{f4_uw}
\end{figure}

\begin{figure}[H]
\centering
	  \includegraphics[scale=.5]{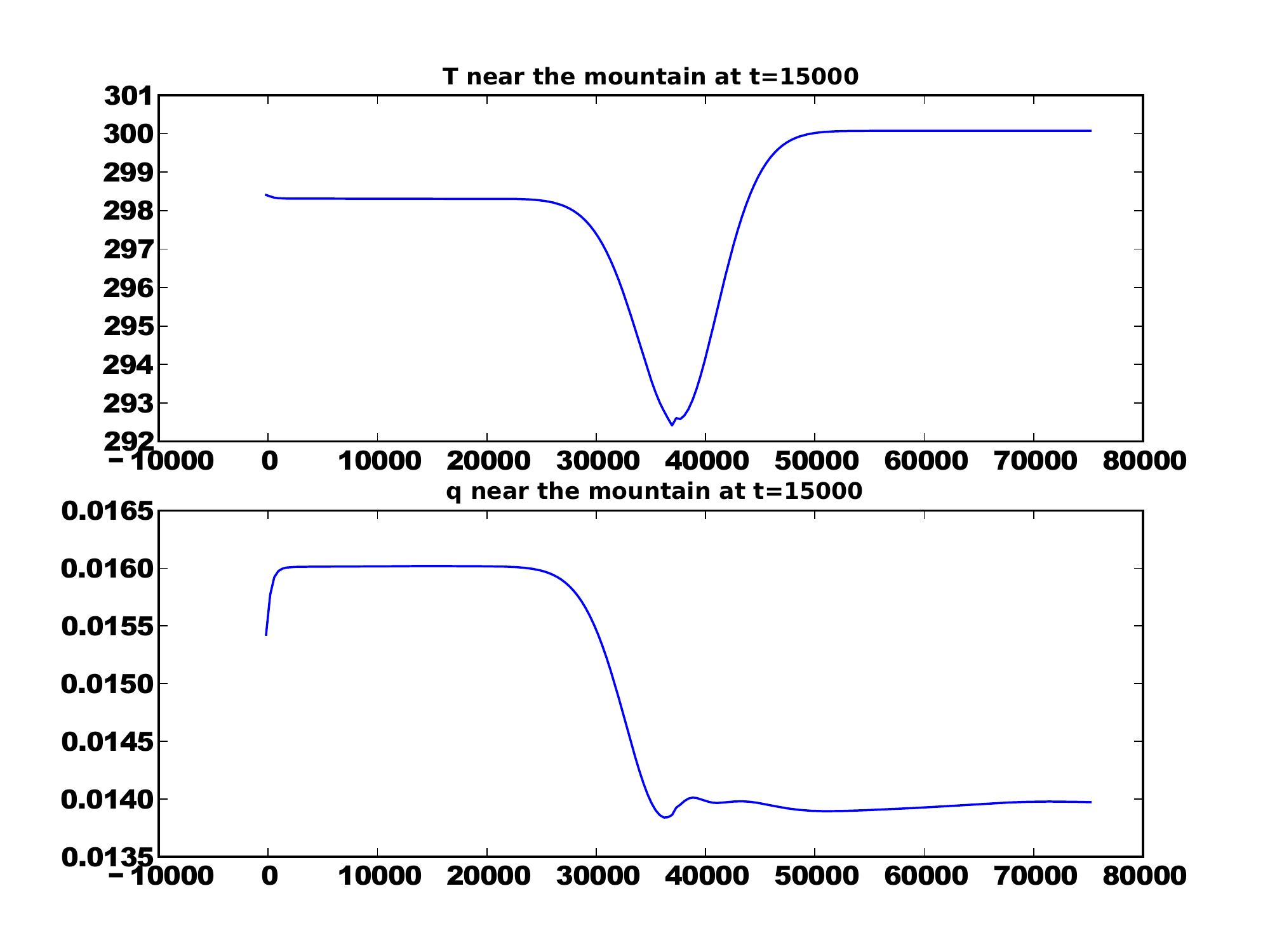}
\caption{1D profile of T and q at t=15000 along the dotted line in Figure \ref{f:slice_1D}.}
\label{f:profile_T_q}
\end{figure}

\begin{figure}[H]
\centering
	  \includegraphics[scale=0.9]{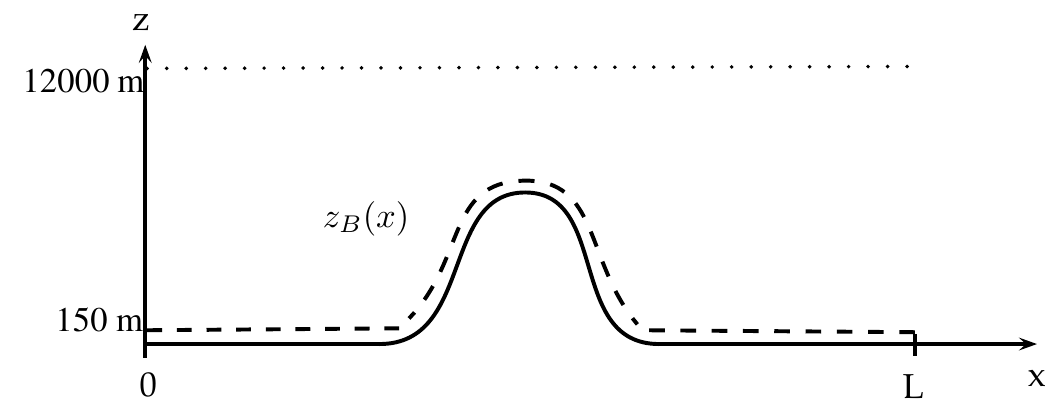}
\caption{Sample path along the mountain for 1D profile used in Figure \ref{f:profile_T_q}}
\label{f:slice_1D}
\end{figure}

\begin{figure}[H]
\centering
	  \includegraphics[scale=0.35]{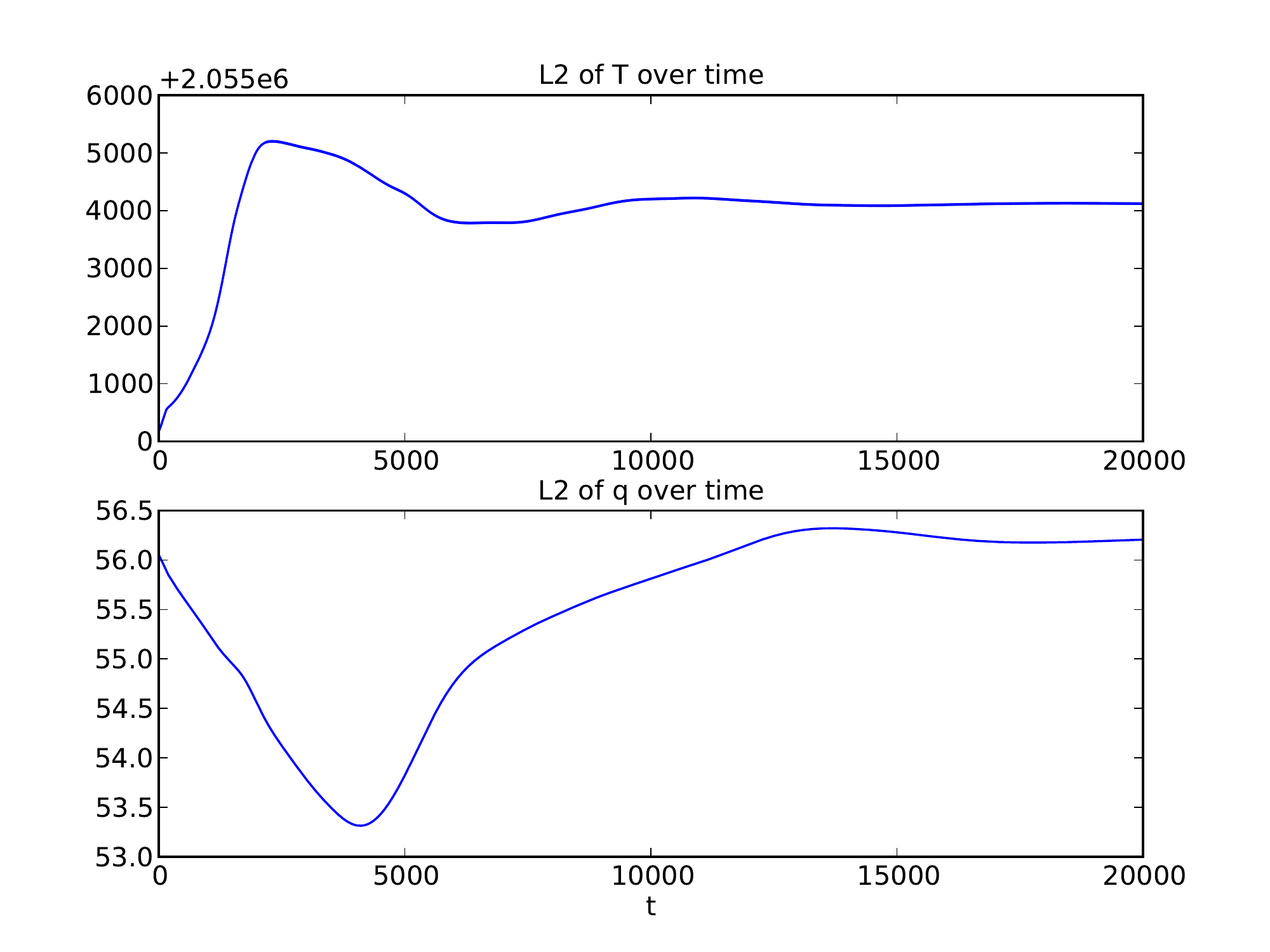}
	  \includegraphics[scale=0.35]{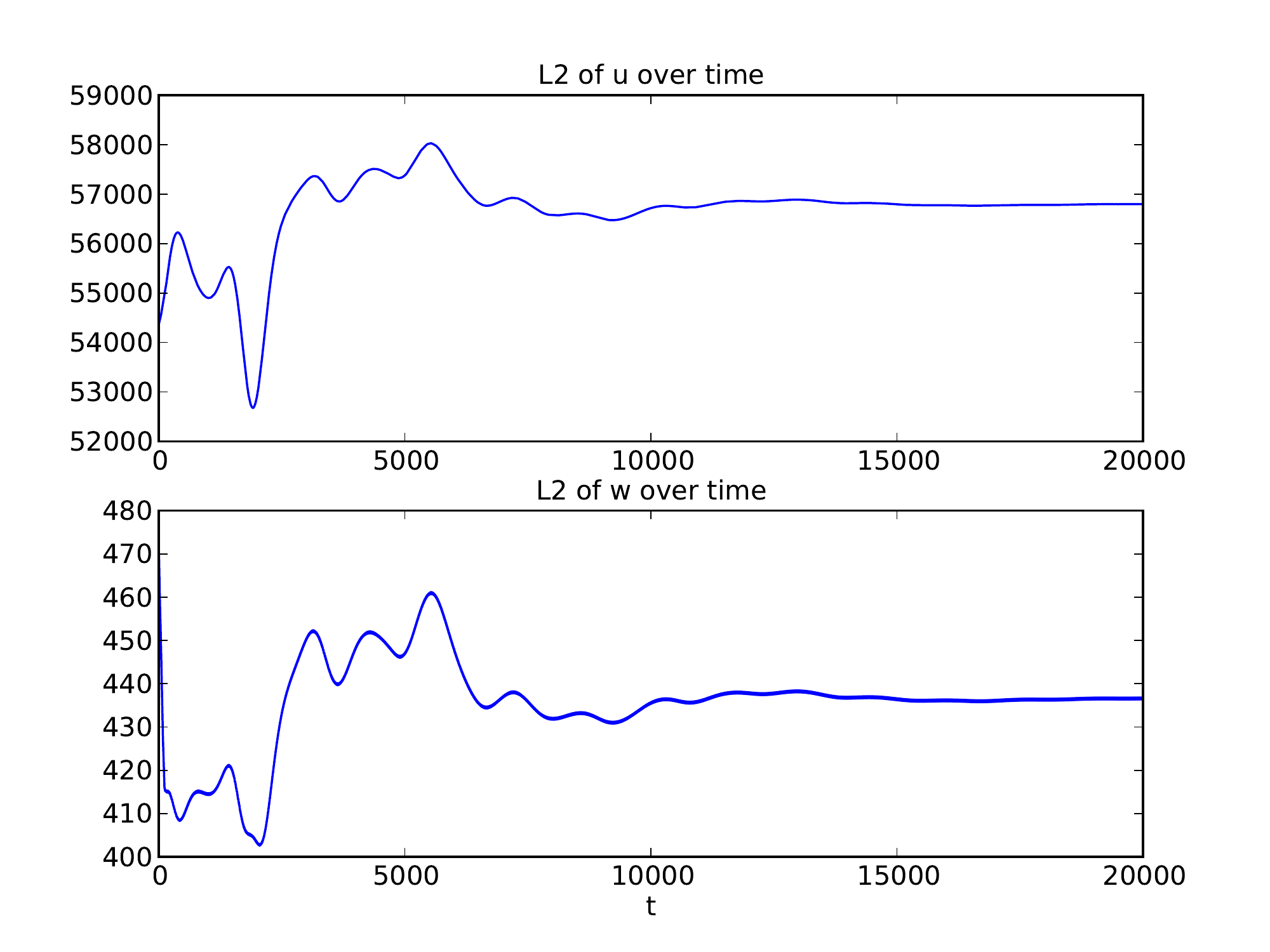}
\caption{{\mk Time-evolution of the (spatial) $L^2$-norm of $T$ and $q$ (on the left), and of $u$ and $\omega$ (on the right).}}
\label{f:L2_energy}
\end{figure}

\begin{figure}[!Hbtp]
      \centering    
      \includegraphics[height=0.7\textwidth, 
width=1\textwidth]{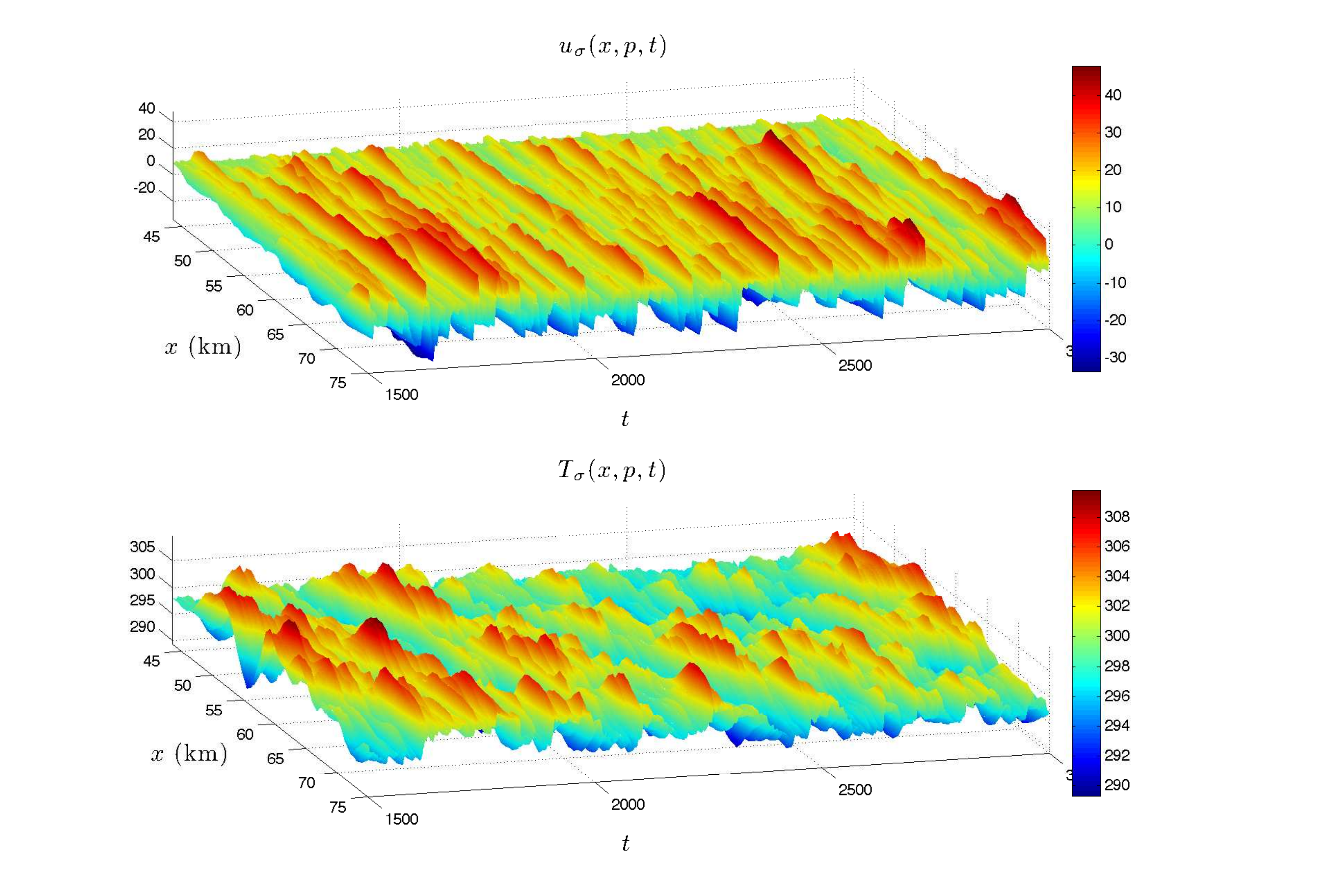}
      \caption{{\footnotesize {\bf Top panel}: Time evolution of the horizontal 
velocity profile  $u_{\sigma}(x,p,t)$ for $\sigma=5$, for $x\in [43.5 \mbox{ 
km}, 75 \mbox{ km}]$ (mainly over which the bombardment occurs) and at a 
``fixed'' value of the topography level.
      {\bf Bottom panel:} Same for the temperature profile $T_{\sigma}(x,p,t)$. 
Both fields exhibit modulated waves traveling eastward. Here $480$ cells  are 
randomly ``bombarded''  according  to \eqref{Eq_noise_form} as time flows, for a 
spatial resolution of the model corresponding to $N_x=N_p=100$. The  figure 
shows that large-scale recurrent patterns arise from such a small-scale random 
forcing.}} 
      \label{fig1_noise}
   \end{figure}

\subsection{Recurrent large-scale patterns from random small-scale forcing} \label{sec4.3}
{\mk As motivated at the end of the previous subsection, we analyze here  from a numerical viewpoint, the effects of a stochastic perturbation to the model formulation. In that respect, we adopted  to stochastically perturb only the  $u$-equation in the inviscid primitive equations considered here, which turned out to be enough to illustrate our purpose.} More precisely,
\begin{equation}
 \mbox{d} u+\Big(u\frac{\de u}{\de x}+\omega \frac{\de u}{\de p} +\phi_x\Big) 
 \mbox{d} t= \sigma \eta_D(x,t),
 \label{e:eq_u_noise}
\end{equation}
where $\eta_D(x,t)$ is a random ``bombardment'' over the region 
$\mathcal{D}\simeq[45\mbox{km},75\mbox{km}]\times[0,2000]$ (at the east side of 
the mountain) which takes the following form: 
\begin{equation}\label{Eq_noise_form}
\eta(x,t)=\sum_{j=1}^N\chi_{B(\mathbf{x}_j(t),r_j)} \mbox{d} W_t.
\end{equation}
Here, $W_t$ denotes a one-dimensional Brownian motion and the $N$ centers 
$\mathbf{x}_j(t)$ of the balls $B(\mathbf{x}_j(t),r_j)$ are drawn uniformly in 
$\mathcal{D}$ as $t$ flows.  In practice the radii $r_j$ are chosen to take a 
fixed value $r$ so that  $B(\mathbf{x}_j(t),r) \subset \mathcal{D}$ for all $t$, 
with $r$ to be characteristic of some small spatial scales for the problem at 
hand (in what follows $r\approx 1$ km).  Such a noise term can be argued to 
model physical processes that are not accounted for by the given equations such 
as for instance, the vortices that would arise on the east side of the mountain 
from a sufficiently large initial horizontal component of an eastward wind.

\begin{figure}[!Hbtp]
      \centering    
      \includegraphics[height=0.6\textwidth, width=1\textwidth]{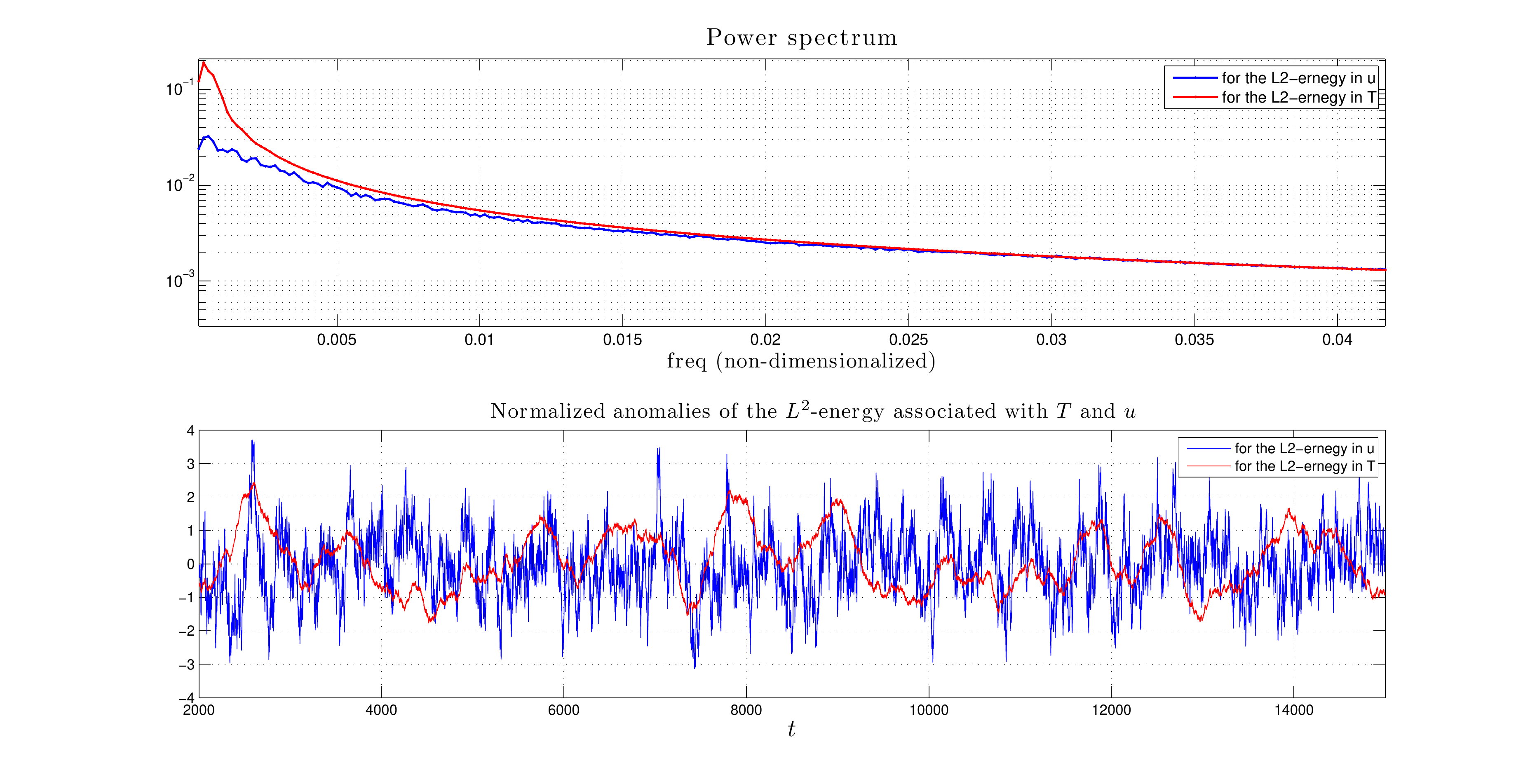}
      \caption{{\footnotesize {\bf Top panel}: Power spectrum (in a 
semilogarithmic scale) of $t\mapsto \|u(\cdot,p,t)\|_{L^2(0,L)}$ (blue curve) 
and of $t\mapsto \|T(\cdot,p,t)\|_{L^2(0,L)}$ (red curve). {\bf Bottom panel}: 
Anomalies of $t\mapsto \|u(\cdot,p,t)\|_{L^2(0,L)}$ and $t\mapsto 
\|T(\cdot,p,t)\|_{L^2(0,L)}$, normalized by their respective (empirical) 
standard deviation.}} 
      \label{fig2_noise}
   \end{figure}   
   
It has been observed that on a spatial resolution of the model corresponding to 
$N_x=N_p=100$,
such a noise term can help trigger interesting dynamics 
such as illustrated in Figure \ref{fig1_noise}. Indeed for $\sigma=0$, the 
system is a stationary regime (steady state) whereas as $\sigma$ starts to increase,  
a complex spatio-temporal dynamics takes place in the $u_{\sigma}$-fields (horizontal velocity with \eqref{e:eq_u_noise}) as well as the $T_{\sigma}$-fields (temperature with \eqref{e:eq_u_noise}); see Figure \ref{fig1_noise} below for $\sigma=5$.\footnote{Less 
interesting dynamics has been observed on the $w$- and $q$-fields which are 
evolving mainly on similar spatio-temporal scales than those of the random 
forcing \eqref{Eq_noise_form} (not shown).}

 \begin{figure}[!Hbtp]
      \centering    
      \includegraphics[height=0.42\textwidth, 
width=.75\textwidth]{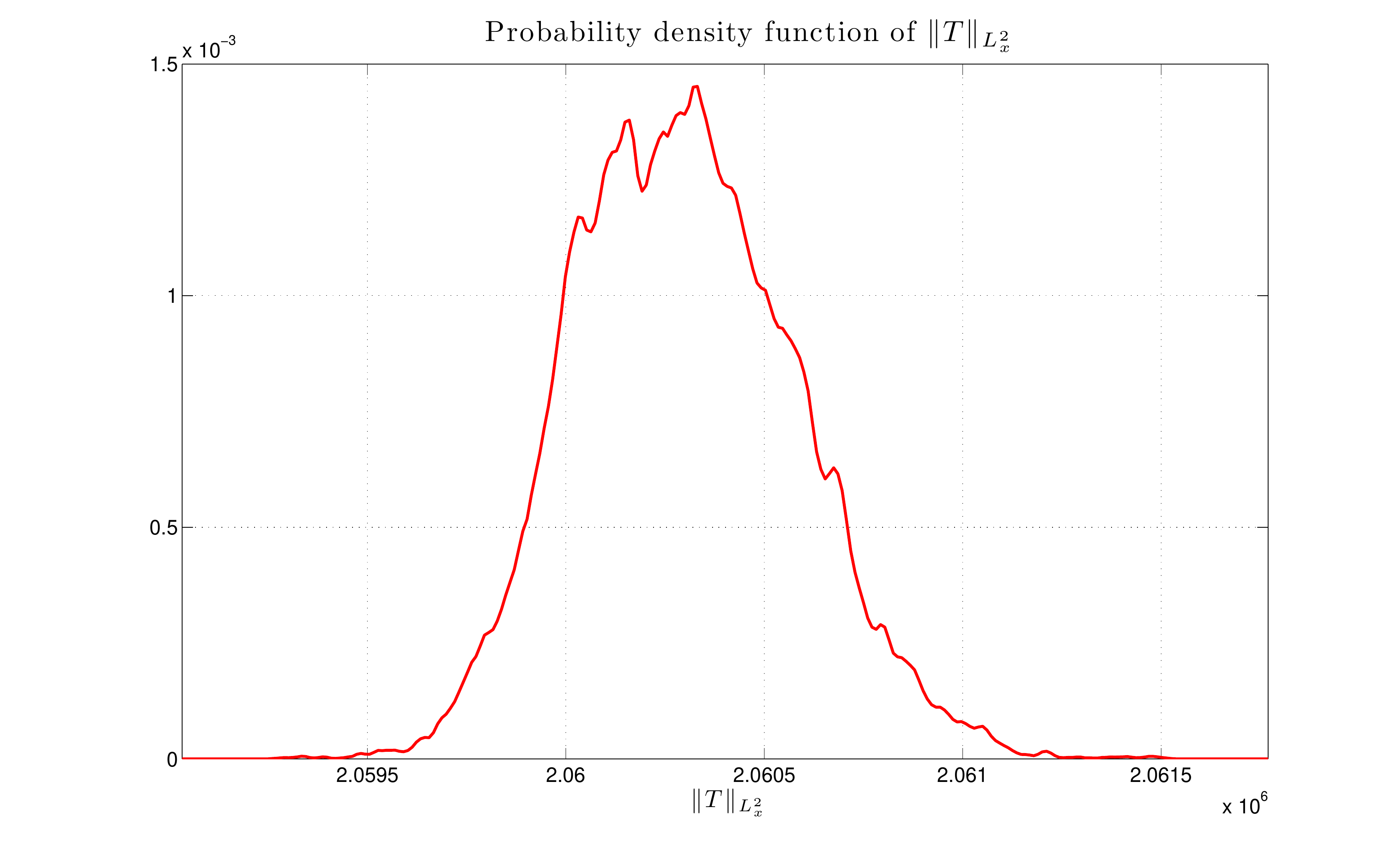}
      \caption{{\footnotesize Empirical probability function of $t\mapsto 
\|T(\cdot,p,t)\|_{L^2(0,L)}$. Non-Gaussian  statistics are observed.}} 
      \label{fig3_noise}
   \end{figure}
 
A closer look at the horizontal velocity and the temperature 
fields\footnote{denoted by $u_{\sigma}$ and $T_{\sigma}$ respectively.} reveals 
that recurrent large-scale patterns, while evolving irregularly in 
time\footnote{and manifesting different characteristics for $u_{\sigma}$ and 
$T_{\sigma}$.}, are now the dominant ingredients of the time evolution of 
these scalar fields. These patterns are mainly expressed here, as waves 
traveling eastward whose amplitude is irregularly modulated as time flows; the 
dominant ``quasi-period'' for the $T_{\sigma}$-field being noticeably larger 
than for the $u_{\sigma}$-field as can be observed on Figure \ref{fig1_noise}.

Without entering in a detailed analysis of the space-time variability of such 
patterns {\mk that could be performed for instance by some} multivariate data-adaptive spectral methods \cite{Ghil_etal02}, 
a simple unidimensional spectral analysis of the time-evolution of the 
$L^2$-energy (in the $x$-direction) contained in the respective fields gives 
already good information about the recurrence characteristics of such 
fields.\footnote{In other words, the choice of the  $L^2$-energy as observable 
allows here to capture key features of the variability of the given 
spatio-temporal fields. We mention that such comments have to be understood 
within the language of the spectral theory of dissipative dynamical systems; see 
 \cite{CNKMG14} for a brief introduction {\mk on the topic}.} In what follows, we will denote 
by $t\mapsto \|u(\cdot,p,t)\|_{L^2(0,L)}$ and  $t\mapsto 
\|T(\cdot,p,t)\|_{L^2(0,L)}$ these respective time-dependent energies. Figure \ref{fig2_noise} below reports on a standard numerical estimation of the power spectrum\footnote{also known as the power spectral density.} associated with the time-variability of these energies, 
such as obtained from their corresponding autocorrelation functions \cite{ER85,Ghil_etal02}. The latter are estimated from the respective anomalies reported on  the bottom panel of Figure \ref{fig2_noise}, after normalization by the standard deviation to plot the curves on a similar order of magnitude.

The numerical results indicate  that the signal $t\mapsto \|T(\cdot,p,t)\|_{L^2(0,L)}$ contains a broadband peak that stands above an exponentially decaying background at low-frequencies within the band $[0,2.5\times10^{-3}]$; see red curve on top panel of Figure \ref{fig2_noise}. 
Associated with the signal $t\mapsto \|u(\cdot,p,t)\|_{L^2(0,L)}$, a broadband peak stands also above  an exponentially decaying background, but with much less energy contained in it and spread over a broader range of frequencies (almost over $[0,5 \times 10^{-3}]$); see blue curve on top panel of Figure \ref{fig2_noise}.

 The fact that the broadband peak associated with the time evolution of\\ 
$\|u(\cdot,p,t)\|_{L^2(0,L)}$ spreads over the range $[2.5\times 
10^{-3},5\times 10^{-3}]$ in the frequency domain, is consistent with the 
higher-frequency time evolution exhibited by the field $u$ compared to the field 
$T$ as can be observed on Figure \ref{fig1_noise} in the space-time domain, as 
well as the higher-frequency  oscillations exhibited by the evolution of 
$\|u(\cdot,p,t)\|_{L^2(0,L)}$ compared to the one of 
$\|T(\cdot,p,t)\|_{L^2(0,L)}$, in the time domain alone; see Figure \ref{fig2_noise} bottom panel.

   \begin{figure}[!Hbtp]
      \centering    
      \includegraphics[height=0.7\textwidth, 
width=1\textwidth]{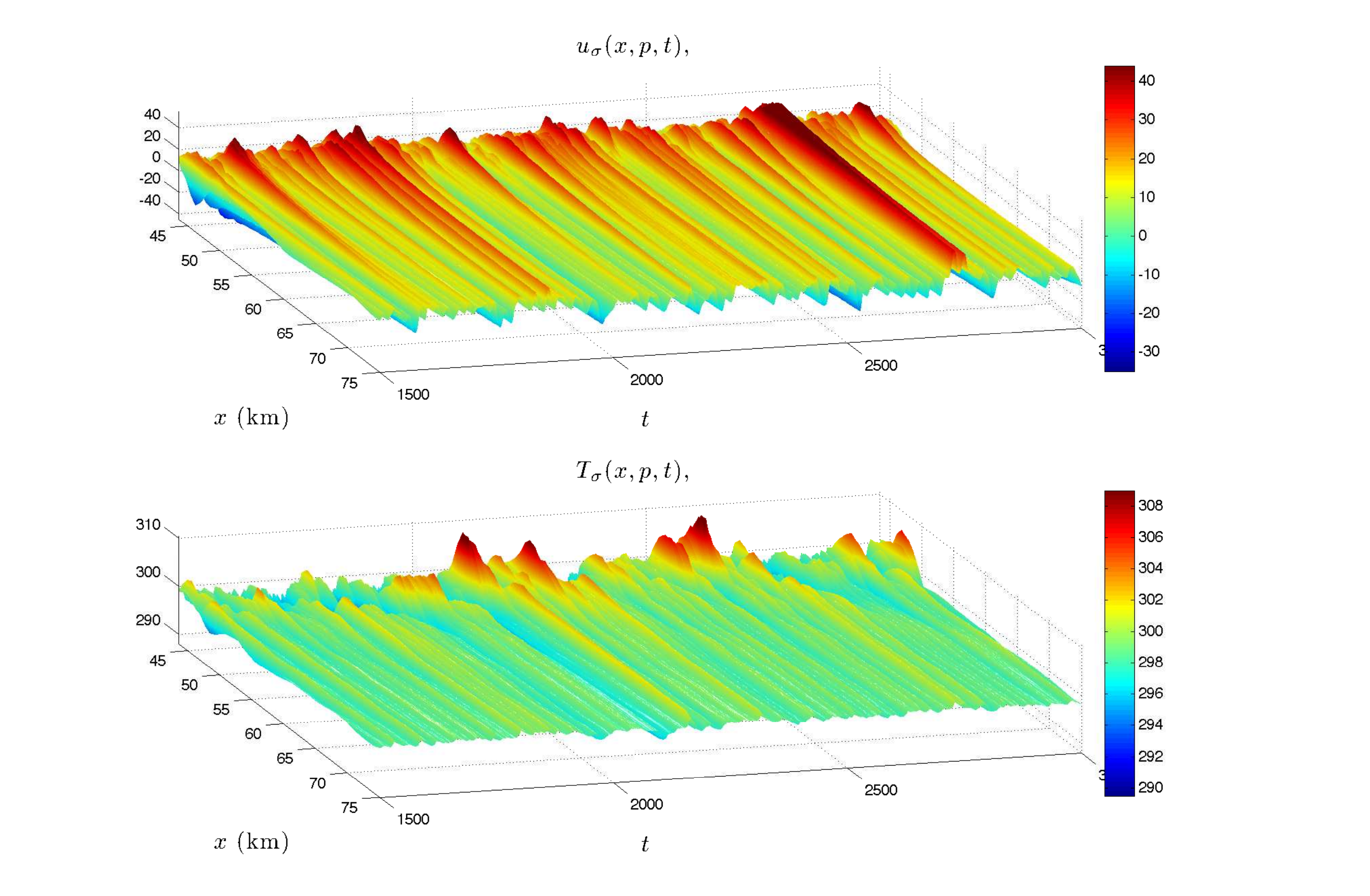}
\caption{{\footnotesize      
Here the topography has been removed (flat domain) while the other parameters of 
the model are kept the same than those used for Figure \ref{fig1_noise}. 
Interestingly, 
these patterns develop much less spatial irregularities as time flows  compared  
to those of Figure \ref{fig1_noise}. The statistics of $t\mapsto 
\|T(\cdot,p,t)\|_{L^2(0,L)}$ are still non-Gaussian (not shown). The latter 
property is here again a signature of some nonlinear effects triggered by the 
noise term \eqref{Eq_noise_form}; nonlinear effects which however give rise to a 
less complex spatial structure of the patterns that develop as time flows, 
compared to the case with topography.}}
\label{figflat_noise}
\end{figure}

The spatio-temporal evolution of $u$ and $T$ exhibit thus recurrent patterns 
that are modulated in time and occur on a multiplicity of scales, whose the 
dominant ones are significantly larger than the spatio-temporal scales on which 
the random forcing act upon. Such a complex dynamical behavior can be argued to 
be consistent with the idea that the noise has triggered some nonlinear effects 
that were not expressed when $\sigma=0$; idea further supported by the non-Gaussian character of the model's dynamics such as observed on the probability density function associated with $t\mapsto \|T(\cdot,p,t)\|_{L^2(0,L)}$; see Figure \ref{fig3_noise}.

Interestingly the characteristics of the large-scale patterns are noticeably different in the case without topography compared to the case with topography, while still resulting from nonlinear effects triggered by the noise; see Figure \ref{figflat_noise}. {\mk To  summarize, the random forcing \eqref{Eq_noise_form}  combined with the numerical scheme developed in this study, allow for a nice illustration of the fact} that the topography is a determining factor for the generation of  complex large-scale patterns in geophysical fluid models.

\section*{\bf Acknowledgements}

This work was supported in part by NSF Grant DMS 1206438, and by the Research Fund of Indiana University.
MDC was partially supported by  the Office of Naval Research grant N00014-12-1-0911.

\end{document}